\documentclass[12pt,a4paper]{article}

\usepackage{fix-cm}

\usepackage[utf8]{inputenc}
\usepackage[T1]{fontenc}

\usepackage{algorithm,algpseudocode}
\usepackage[many]{tcolorbox}
\newcounter{algoBoxCounter}
\newtcolorbox[use counter=algoBoxCounter]{algoBox}[2]{
  colback=gray!5,          
  colframe=gray!80,        
  enhanced,
  drop shadow={black!50},
  boxrule=1pt,
  float=t!,
  title={Algorithm~\thealgoBoxCounter: #1},  
  width=\textwidth,
  before skip=0pt,
  after skip=0pt,
  label type=algoBoxCounter,
  label={#2}               
}
\usepackage{tikz}
\usetikzlibrary{arrows.meta, positioning, calc}

\usepackage[top=1in, bottom=1in, left=1in, right=1in]{geometry}
\usepackage{setspace}
\usepackage{titlesec}
\usepackage{pdflscape}
\usepackage{graphicx}
\usepackage{booktabs}
\titlespacing*{\section}{0pt}{1ex}{2ex}
\titlespacing*{\subsection}{0pt}{2ex}{1ex}
\titlespacing*{\subsubsection}{0pt}{1ex}{1ex}
\usepackage{afterpage}

\usepackage{amsthm,bm,amsmath,amsfonts,amssymb,newtxtext,newtxmath} 
\DeclareMathOperator{\plim}{plim}

\makeatletter
\g@addto@macro\normalsize{%
  \setlength\abovedisplayskip{5pt plus 2pt minus 2pt}%
  \setlength\belowdisplayskip{5pt plus 2pt minus 2pt}%
  \setlength\abovedisplayshortskip{6pt plus 2pt minus 2pt}%
  \setlength\belowdisplayshortskip{6pt plus 2pt minus 2pt}%
}
\makeatother
\allowdisplaybreaks
\theoremstyle{plain}
\newtheorem{theorem}{Theorem}
\newtheorem{proposition}{Proposition}
\newtheorem{corollary}{Corollary}
\newtheorem{assumption}{Assumption}
\newtheorem{lemma}{Lemma}

\theoremstyle{definition}

\theoremstyle{remark}
\makeatletter
\renewenvironment{proof}[1][\proofname]{\par
    \pushQED{\qed}%
    \normalfont 
    \topsep 12pt plus 4pt minus 3pt
    \trivlist
    \item\relax
          {\bfseries
      #1\@addpunct{.}}\hspace\labelsep\ignorespaces
  }{%
    \popQED\endtrivlist\@endpefalse
    \addvspace{13pt plus 4pt minus 3pt}
  }
\makeatother

\usepackage{graphicx}
\usepackage[section]{placeins}
\usepackage{threeparttable}
\usepackage{siunitx}

\usepackage{hyperref}
\hypersetup{
  colorlinks=true,
  linkcolor=[RGB]{139,0,0},     
  citecolor=[RGB]{0,0,139},     
  urlcolor=[RGB]{0,50,150}
}
\usepackage{appendix}

\usepackage[comma]{natbib}
\usepackage{cleveref}
\crefname{figure}{figure}{figures}
\creflabelformat{figure}{#2#1#3}
\crefname{equation}{equation}{equations}
\creflabelformat{equation}{#2#1#3}
\crefname{lemma}{lemma}{lemmas}
\creflabelformat{lemma}{#2#1#3}
\crefname{theorem}{theorem}{theorems}
\creflabelformat{theorem}{#2#1#3}
\crefname{condition}{condition}{conditions}
\creflabelformat{condition}{#2#1#3}
\crefname{assumption}{assumption}{assumptions}
\creflabelformat{assumption}{#2#1#3}
\crefname{appendix}{appendix}{appendices}
\creflabelformat{appendix}{#2#1#3}
\crefname{algoBox}{algorithm}{algorithms}
\creflabelformat{algorithm}{#2#1#3}
\crefname{enumi}{}{}
\creflabelformat{enumi}{(#2#1#3)}
\crefname{appsec}{Appendix}{Appendices}
\creflabelformat{appsec}{#2#1#3}

\usepackage{enumitem}

\begin{document}

\title{\Large Social Interactions Models with Latent Structures
}

\date{}

\author{Zhongjian Lin\footnote{University of Georgia, \href{mailto:zhongjian.lin@uga.edu}{zhongjian.lin@uga.edu}} \and Zhentao Shi\footnote{The Chinese University of Hong Kong, \href{mailto:zhentao.shi@cuhk.edu.hk}{zhentao.shi@cuhk.edu.hk}}\and Yapeng Zheng\footnote{The Chinese University of Hong Kong, \href{mailto:yapengzheng@link.cuhk.edu.hk}{yapengzheng@link.cuhk.edu.hk}}}

\onehalfspacing

\maketitle

\begin{abstract}
This paper studies estimation and inference of heterogeneous social interactions featuring group fixed effects and slope heterogeneity under latent structure. Theoretically, in a binary choice model with social interactions we characterize the asymptotic bias of the slope coefficient due to the presence of high dimensional incidental parameters. Computationally, we invoke parametric bootstrap method to debias and establish asymptotic validity. Furthermore, tailored algorithms are employed to explore the latent structures and determine the number of clusters. Monte Carlo simulations confirm strong finite sample performance of our methods. In an application to students' risky behaviors, the algorithm detects two latent clusters and finds that peer effects are significant within one of the clusters, demonstrating the practical applicability in uncovering heterogeneous social interactions.

\bigskip

\noindent {\bf Keywords:} Classification, C-Lasso, Heterogeneity, Penalized Profile Likelihood, Sequential Estimation, Social Interactions

\medskip

\noindent {\bf JEL Classifications}: C35, C38, C57, I21.
\end{abstract}

\newpage

\onehalfspacing

\section{Introduction}
\noindent Binary choice models with social interactions, or peer effects, are widely employed to study decision-making processes among people connected through social networks. A typical empirical setting involves individuals distributed across many distinct groups (e.g., students assigned to classrooms with different class sizes and teaching qualities \citealp{graham2008identifying,burke2013classroom,sacerdote2011peer}), where interactions occur within groups only. When estimating peer effects, a naive approach to utilizing all available data would assume homogeneous parameters (including the intercept and slopes) across groups \citep{brock2001discrete}.
However, the homogeneity assumption, if violated, may incur severe bias. Conversely, allowing complete heterogeneity often yields unreliable group-level estimates due to small sample sizes. 

To balance the trade-off between model misspecification risk and statistical efficiency, we propose a binary choice model featuring social interactions, group fixed effects, and latent structures. 
In our empirical application with the National Longitudinal Study of Adolescent to Adult Health (Add Health) dataset, this model enables us to analyze heterogeneous peer effects among students across different clusters of schools. Our proposed method simultaneously identifies the latent clusters and estimates the unknown parameters, revealing that students in one of the clusters of schools exhibit statistically significant peer effects.

The binary choice model with social interactions and a latent structure is analyzed through a simultaneous game framework, where agents interact on exogenously given social networks within groups.   The exogenous network assumption is standard in peer effects studies. The linear-in-means model adopts it \citep{manski1993identification,manski2000economic},
as does the discrete choice model \citep{brock2001discrete,lin2017estimation}. 
The nested pseudo likelihood (NPL, \citealp{aguirregabiria2007sequential,lin2017estimation}) algorithm is the off-the-shelf method for estimating the structural parameters in games. The NPL algorithm sequentially estimates unknown parameters and updates the conditional choice probabilities (CCPs), thereby circumventing the computational challenges related to the growing dimension of states and actions. 

The data structure under consideration is similar to panel data, which is ubiquitous in econometric analysis.
Because panel data track individuals over time, the rich information they contain can produce more efficient estimators than those based solely on cross-sectional or time-series approaches. While parameter homogeneity is typically assumed to leverage the advantages of panel data --- including enhanced cross-sectional averaging power --- empirical studies frequently reject this assumption, see \citet{phillips2007transition}, \citet{browning2007heterogeneity}, \citet{su2013testing}, and \citet{lu2017determining}. As \citet{hsiao2022analysis} notes in his Chapter 1, inferences based on invalid homogeneity assumptions may be misleading. 
On the other end of the spectrum of specifications is complete parameter heterogeneity across all groups. However, the limited sample sizes associated with each group yield less reliable estimators.

This dilemma motivates us to examine an intermediate case through a latent structure model, where parameters are heterogeneous across clusters but homogeneous within clusters, with neither the number of clusters nor cluster membership known \emph{a priori}. Many econometric papers on social interactions hinge on pre-specified cluster structures. For example, \cite*{lewbel2023social} analyze social effects among elementary school students in the Student/Teacher Achievement Ratio (STAR) Project in the U.S. State of Tennessee. They partition the classes (the groups) into two environments (clusters) according to class sizes and estimate peer effects for each environment, respectively (this small/large classification is also adopted by \citet{graham2008identifying} to identify peer effects). However, external variables for identifying the latent structure are not always available. To address this problem, we integrate the Classifier-Lasso (C-Lasso, \citealp*{su2016identifying}) algorithm with the NPL estimation. The C-Lasso procedure solves a penalized optimization problem where a multiplicative penalty term is appended to a primary loss function. We show that it can identify the latent structure, thereby enabling analysis of peer effects within each cluster. 

While the C-Lasso procedure constitutes a non-convex problem approximable by a sequence of convex problems, its computational demands remain substantial. Moreover, the stability of the NPL algorithm may suffer when CCP updating becomes complex. We therefore propose a multi-step estimation algorithm that decouples C-Lasso from the NPL steps: (i) apply NPL to obtain preliminary parameter and CCP estimates for each group; (ii) plug the estimated CCPs into C-Lasso to extract cluster classifications; (iii) implement NPL to obtain post-classification estimators for each cluster. We establish large-sample results for all steps, including classification consistency for the membership and asymptotic normality of the post-classification estimators. Since group fixed effects are jointly estimated, an incidental parameter bias \citep{neyman1948consistent,hahn2004jackknife} affects the estimators of the slope parameters.
To our knowledge, no paper has studied this problem in the panel data setting for the NPL estimators.
We formally characterize this bias in an asymptotic framework where both the number of groups and group size grow to infinity while their ratio remains bounded. The popular split-panel jackknife estimator \citep{dhaene2015split,li2025baggingnetwork,mei2026nickell} and other split-sample-based bias-correction methods are difficult to implement in our framework because CCPs are jointly determined by all individuals in a group, preventing parameter estimation using only a subset of observations. Moreover, due to the presence of social effects, the analytical bias is too complex to estimate via simple plug-in methods.  We employ the bootstrap inference method of \citet{higgins2024bootstrap} for panel fixed-effect models to construct confidence intervals with correct coverage levels.

Beyond inference on the slope parameters, we develop inference on the average partial effects (APEs), a policy-relevant summary of a binary choice model, defined as the average change in the equilibrium choice probability induced by a unit perturbation of a covariate. In our setting, APEs are not simply the average of the partial derivatives to the covariates. Through the equilibrium fixed point, the influence of any covariate is amplified by a \emph{social multiplier} that depends on the within-group network and the peer-effect strength, so the APE captures both the direct effect of a covariate and its network-mediated indirect effect. We provide cluster-level APEs as well as a population-weighted APE that aggregates across clusters, with confidence intervals delivered by the parametric bootstrap. 

To summarize, our contributions are twofold. First, in the classical setting, individuals across groups are homogeneous.
This paper allows for the presence of fixed effects to capture rich heterogeneity.
Our theoretical contribution lies in
revealing the bias due to the high-dimensional incidental parameters
in a model with a ``fixed point'', which was unknown to the literature of endogenous peer effects,
to the best of our knowledge.
Second, as engineering contributions, we adapt bootstrap methods from the panel
data literature to overcome the above-mentioned incidental parameter bias, and further
enhance the flexibility of panel data modeling of peer effects by accommodating cluster-level heterogeneity via latent structures.
From the practical side, 
we provide a pipeline of tailored algorithms to detect the latent structures, estimate the structural parameters, and conduct inference about objects of 
economic interest. 
This paper aligns with the ongoing movement in relaxing classical parsimonious econometric models to meet the challenges brought about by big data, large models, and real-world complexities.

\subsection{Related literature}
Our paper contributes to the literature on social interactions models. \citet{manski1993identification} introduced the famous linear-in-means model and demonstrated the non-identification of peer effects due to the ``reflection problem''. Subsequent studies have extended and refined this framework, including \cite*{lee2007identification,graham2008identifying,bramoulle2009identification,calvo2009peer,goldsmith2013social,ross2022measuring,lin2026social}, among others.
We focus on binary choice models, which were pioneered by \citet{brock2001discrete,brock2001interactions}. They established identification for binary choices in social interactions models where each group member's decision is influenced by a rational expectation of the average choice in the group. \cite*{lee2014binary} further advanced this literature by achieving identification and providing maximum likelihood estimation for a general network model with heterogeneous rational expectations. \cite*{lin2017estimation,lin2021uncovering} and \citet{lin2024binary} investigated the estimation and inference of social interaction effects in various scenarios and proposed using NPL to reduce the computational burden of estimating CCPs. \citet{lin2025endogenous} studied the endogenous treatment models with social interactions in both the treatment and outcome equations. We introduce the group fixed effects, and apply a parametric bootstrap to conduct statistical inference, which fills a gap between the literature on peer effects and panel data. Moreover, the use of latent structure modeling in this context represents a novel contribution to the strand of literature on social interactions. Heterogeneity is ubiquitous in empirical studies. In peer effects studies, researchers find that the dependence of peer effects on a variety of background characteristics \citep*[e.g.][find that the impact of peer weight is larger among females and adolescents with high BMI]{trogdon2008peer}. The latent structure helps us identify heterogeneous peer effects in different clusters without prior knowledge about memberships.

We also contribute to the study of fixed effects models with a latent structure. We introduce the latent structure to social interactions models, and study the classification consistency under an NPL estimation framework. The literature offers several approaches for determining an unknown cluster structure. The first approach relies on finite mixture models. For panel discrete choice models, \citet{kasahara2009nonparametric} and \citet{browning2014dynamic} investigated identification for a fixed number of clusters using nonparametric mixture distributions. In contrast, our model adopts a fully parametric approach and identifies the latent structure through the C-Lasso algorithm. This algorithm has been further applied in other settings by \citet{su2018identifying} and \cite*{su2019sieve}. The second approach is based on the K-means algorithm, commonly used in statistical analysis. \citet{lin2012estimation} and \citet{sarafidis2015partially} applied it to linear panel data models where the slope coefficients exhibit cluster structures. \cite*{bonhomme2015grouped} extended this method to linear panel data models with additive fixed effects having cluster structures and studied the asymptotic properties of the estimator. This approach was further generalized in \cite*{bonhomme2022discretizing} to analyze general fixed effects models. \citet{wang2021identifying} provided a comparison between the C-Lasso and K-means approaches in their introduction, highlighting the relative strengths and limitations of each method. 

This paper proceeds as follows. In Section \ref{sec:Model}, we introduce the heterogeneous binary choice model with latent cluster structures and social interactions and propose the estimation algorithm which combines NPL and C-Lasso. Section \ref{sec:large_sample_results} establishes asymptotic properties of the proposed method. We conduct Monte Carlo experiments in Section \ref{sec:MC} to illustrate the finite sample performance in classification consistency and the quality of the estimates.
Using the Add Health dataset, Section \ref{sec:empirical_application} presents the empirical results on the peer effects among students across different schools in their choices toward risky behaviors. The last section concludes. We present proofs in Appendices \ref{appendix} and \ref{appendix-lemmas}.
\medskip\medskip


\section{Social Interactions Models}\label{sec:Model}

A social interactions model with unobserved group heterogeneity naturally resembles a panel data model.
In standard panel data, the textbook default is ``large cross section and time dimension (usually denoted as $N$ and $T$, respectively)''.  
In the social interactions model, a set of $n$ individuals, for example the students in a school, forms a group, and there are $G$ such groups. Therefore, in our context $G$ is the counterpart of panel data's $N$, while $n$ is the counterpart of panel data's $T$. 
Let $[n]$ denote $\{1,\cdots,n\}$. For simplicity, we present the mathematical formulation as a balanced panel data structure: there are $G$ independent groups, and each group has $n$ individuals interacting with each other. An extension to unbalanced panels only requires a modification of notation.

The groups (e.g., schools) are observed by the researcher, whereas the latent clusters introduced in Section~\ref{sec:adjusted-algorithm} are unobserved and must be estimated from the data. 

A generic individual $i$ in the $g$th group, with a vector of observable characteristics $X_{ig}$, takes a binary action $Y_{ig} \in \{ 0,1\} $. His behavior is influenced by the social network in which he lives, denoted by the friendship adjacency matrix $F^g=\{F^g_{ij}\}_{i,j \in [n]}$. If $F^g_{ij}=1$, then $i$ is influenced by $j$, and $F^g_{ij} = 0$ otherwise. By convention, we set the self-links $F^g_{ii}=0$ for $i\in[n]$. Denote $F_i^g= \{j: F^g_{ij}=1\}$ as the set of $i$'s influencers (or friends), and $N_i^g=\#F_i^g$ as the number of influencers. 

Denote $\mathbb I_g = \{X_{ig},F_i^g\}_{i=1}^n$ as the public information set. We consider discrete games with incomplete information. 

Different individuals have different numbers of peers. To measure the influence of peers' actions on $i$, we collapse it into an average
$$\overline{P}_{ig0}=1\{N_i^g>0\} \cdot \frac{1}{N_i^g} \sum_{j \in F_i^g}P_{jg0},$$
where $P_{jg0}= \Pr(Y_{jg}=1|\mathbb I_g,\varepsilon_{ig})=\mathbb E(Y_{jg}|\mathbb I_g,\varepsilon_{ig})$ is the conditional belief of the influencer's choice. It can be simplified to $P_{jg0}= \Pr(Y_{jg}=1|\mathbb I_g)=\mathbb E(Y_{jg}|\mathbb I_g)$ with the independence assumption on the private information we make below. 
The observed binary outcome $Y_{ig}$ is generated from a random utility model of social interactions: 
\begin{equation}
\label{eq:model}
Y_{ig}=1 \left\{\overline{P}_{ig0}\overline\beta_{g0}+X_{ig}'\beta_{g0}+\mu_{g0}>\varepsilon_{ig} \right \},\quad i\in[n],g\in[G],
\end{equation}
where in the general case the slope coefficient of the peer effect $\overline\beta_{g0}$, the slope coefficient of the individual characteristics $\beta_{g0}$, and the group fixed effect $\mu_{g0}$ can all vary across groups. In dynamic games, \citet{aguirregabiria2007sequential} investigate the unobserved group/market heterogeneity and they treat it as a random effect and a part of the public information. \citet{arcidiacono2012estimating} study spillover effects across groups with fixed effects on continuous outcomes. Here we take $\mu$'s as fixed-effect parameters in the binary choice model. In particular, the fixed effect $\mu_{g0}$ enters the single index additively. 
Although fixed effects are standard in the panel data literature, to the best of our knowledge,
no paper considers fixed effects when estimating peer effects. We provide the first formal results about conducting valid statistical inference in this setting.

 Model \eqref{eq:model} adopts an incomplete information structure for the discrete game, where the belief about influencers' choices, i.e., the conditional expectation, enters the utility function of individuals.  The random error $\varepsilon_{ig}$ is independently and identically distributed (i.i.d.)~across $i\in[n]$ and $g\in[G]$, with a known distribution. In this paper, we assume $\varepsilon_{ig}$ follows a standard logistic distribution; there is no difficulty in allowing other link functions, for example, the probit. Another straightforward extension is to include the contextual effects into the single index, as in \cite*{lin2021uncovering}. We focus on the current simple model to fix ideas.

To set up the building blocks for the estimation scheme, denote the data as $w_{ig}=(Y_{ig},X_{ig}')'$ and the slope coefficient vector as $\theta_g=(\overline\beta_g,\beta_g')'$, with dimension $1+\mathrm{dim}(X_{ig})$. 
The \emph{negative} log-likelihood of an individual $i$ in group $g$ is 
\begin{align*}
   \psi(w_{ig},P_g;\theta_g,\mu_g) 
  =  &-Y_{ig}\log \Lambda(\overline{P}_{ig}\overline\beta_g+X_{ig}'\beta_g+\mu_g) \\
  & -(1-Y_{ig})\log [1-\Lambda(\overline{P}_{ig}\overline\beta_g+X_{ig}'\beta_g+\mu_g)],
\end{align*}
where $\Lambda(x)=(1+\mathrm{e}^{-x})^{-1}$ is the cumulative distribution function (CDF) of the standard logistic distribution.
Let $P_{g}=\{P_{ig}\}_{i\in[n]}$ be a choice probabilities profile, 
and then the average negative log-likelihood for the group $g$ can be written as
\begin{equation*}
    \Psi_{ng}(\mu_g,\theta_g;P_g)=\frac{1}{n}\sum_{i\in[n]}\psi(w_{ig},P_g;\theta_g,\mu_g).
\end{equation*}
Furthermore, the average negative log-likelihood over the $G$ groups is 
\begin{equation}\label{eq:likelihood}
    \Psi_{n}(\bm\mu,\bm\theta;\bm{P})=\frac{1}{G}\sum_{g\in[G]}\Psi_{ng}(\mu_g,\theta_g;P_g),
\end{equation}
where $\bm\mu=\{\mu_g\}_{g\in[G]}$, $\bm\theta=\{\theta_g\}_{g\in[G]}$, and $\bm{P}=\{P_g\}_{g\in[G]}$ concatenate the parameters over all groups.

\subsection{Nested Pseudo Likelihood Algorithm: Common Slope Parameter}

In empirical work using fixed effects models, most applications presume homogeneous slope coefficients to capture the commonality across groups while allowing the fixed effects to be unrestricted. We start with this convention by restricting $$\theta_1=\cdots=\theta_{G}=\theta=(\overline\beta,\beta')',$$ 
under which the negative log-likelihood function $\Psi_{n}(\bm\mu,\bm\theta;\bm{P})$ in \eqref{eq:likelihood} can be written as $\Psi_n(\bm{\mu},\theta;\bm{P})$, where a common $\theta$ is shared by all groups.
 
Given the model \eqref{eq:model}, the equilibrium CCP profile $\bm{P}(\bm\mu,\theta)=\{P_{ig}(\mu_g,\theta)\}_{i\in[n],g\in[G]}$ is determined as fixed points of the following equations
\begin{equation}
\label{eq:CCP-def}
P_{ig}=\Lambda(\overline{P}_{ig}\overline\beta+X_{ig}'\beta+\mu_g),\  i\in[n],g\in[G].
\end{equation} 
Let $\mathcal{A} \subset \mathbb R$ be a common parameter space for all fixed effects $\mu_g$
and $\Theta  \subset \mathbb{R}^{1+\mathrm{dim}(X_{ig})} $ be the common parameter space for $\theta$ --- both are assumed to be compact. 
Algorithm \ref{alg:NPL_alg_FE} presents a sequential estimation procedure for fixed effects and common parameters jointly via the NPL algorithm. In this algorithm, parameter estimation and CCP updates are executed sequentially, which simplifies the computation relative to nonparametric estimation of CCPs   \citep[see][]{aguirregabiria2002swapping,aguirregabiria2007sequential,lin2017estimation,lin2024binary}.

\begin{algoBox}{NPL with Fixed Effects}{alg:NPL_alg_FE}
\begin{algorithmic}[1]
\small
\State Initialize $P_{ig}^{(0)}$, $i\in[n],g\in[G]$, and set $t = 1$
\Repeat
    \State $(\bm\mu^{(t)},\theta^{(t)}) \gets \underset{(\bm\mu,\theta)\in\mathcal{A}^G\times\Theta}{\arg\min} \Psi_{n}(\bm\mu,\theta; \bm{P}^{(t-1)})$ \Comment{Minimize negative likelihood}
    \State $P_{ig}^{(t)} \gets \Lambda\Big(X_{ig}'\beta^{(t)}+\mu_g^{(t)}+\overline{P}_{ig}^{(t-1)}\overline\beta^{(t)}\Big),\ i\in[n],g\in[G]$ \Comment{Update CCPs}
    \State $t \gets t + 1$
\Until{$\max_{i\in[n],g\in[G]}|P_{ig}^{(t)} - P_{ig}^{(t-1)}| \leq 10^{-5}$}
\State \textbf{Output}: $(\widehat{\bm\mu},\widehat\theta)= (\bm\mu^{(t)},\theta^{(t)})$ and $\widehat{\bm{P}} = \widehat{\bm{P}}^{(t)}$
\end{algorithmic}
\end{algoBox}

We will establish the asymptotic normality of $\widehat\theta$ in Theorem \ref{thm:post_classification_estimation}\ref{thm:post_classification_estimation-a}. However, statistical inference based on this distributional result faces two difficulties. First, similar to fixed effects estimators in the panel data literature, $\widehat\theta$ suffers from an incidental parameter problem \citep{neyman1948consistent}. Theoretical analyses in Appendix \ref{appendix} show that the analytic form of the bias term depends not only on third derivatives of the likelihood function but also on the derivatives of the CCPs $\bm{P}(\bm\mu,\theta)$, indicating that the analytic bias-correction is too complicated to carry out. Moreover, as the CCPs are jointly determined through social interactions, split-sample-based bias-correction methods \citep{dhaene2015split} do not work. Second, direct estimation of asymptotic variance performs poorly in practice. Therefore, the literature on estimating peer effects relies on (nonparametric) bootstrap to conduct statistical inference \citep*[e.g.][]{lewbel2023social}. Such a method can only provide estimates of standard errors, while the bootstrap confidence intervals are invalid due to the asymptotic bias. 

While Algorithm \ref{alg:NPL_alg_FE} provides point estimates of the parameters, to conduct statistical inference about the slope coefficient $\theta$, we proceed with Algorithm \ref{alg:Bootstrap}, a version of the \emph{parametric} bootstrap method for fixed effects models proposed by \citet{higgins2024bootstrap} adapted to the current setting. 

\begin{algoBox}{Bootstrap Inference for $a'\theta$}{alg:Bootstrap}
\begin{algorithmic}[1]
\small
\State \textbf{Input:} A non-random vector $a$, a level $\alpha$, and the number of replications $B$
\State \textbf{Initialize:} Get initial estimator $\widehat{\theta}$ from Algorithm \ref{alg:NPL_alg_FE}
\For{$b = 1$ to $B$}
    \State Draw $\varepsilon_{ig}^b \sim \text{Logistic}(0,1)$ independently
    \State Generate $Y_{ig}^b = 1\{\widehat{\overline{P}}_{ig}\widehat{\overline{\beta}} + X_{ig}'\widehat{\beta} + \widehat{\mu}_g > \varepsilon_{ig}^b\}$
    \State Obtain $\widehat{\theta}^b$ via Algorithm \ref{alg:NPL_alg_FE} using $(Y^b, X)$ \Comment{Bootstrap}
\EndFor
\State Compute empirical quantile function $Q^*$ of $\{a'(\widehat{\theta}^b-\widehat\theta)\}_{b\in[B]}$
\State \textbf{Output:} 
\State \quad Debiased estimator: $a'\widehat{\theta}^* = a'\widehat{\theta} - Q^*(1/2)$ \Comment{Estimation}
\State \quad Confidence interval: $[a'\widehat{\theta} - Q^*(1-\alpha/2),\ a'\widehat{\theta} - Q^*(\alpha/2)]$ \Comment{Inference}
\end{algorithmic}
\end{algoBox}

In Algorithm~\ref{alg:Bootstrap}, $\widehat{\theta}^b$ is the NPL estimator from the $b$th bootstrap sample, $Q^*$ is the empirical quantile function of $\{a'(\widehat{\theta}^b - \widehat{\theta})\}_{b=1}^B$, and $a'\widehat{\theta}^* = a'\widehat{\theta} - Q^*(1/2)$ is the debiased scalar estimator.  The confidence interval inverts the bootstrap distribution by subtracting the appropriate quantiles from $a'\widehat{\theta}$.

\subsection{Classification and Estimation with Latent Structures}
The previous subsection assumes that the $\theta_{g0}$ are common for all groups, leading to the standard discrete game model without group-level slope heterogeneity. However, peers' influence may vary across different schools; for example, \cite*{lewbel2023social} document different levels of peer effects between large and small groups. Therefore, neglecting the slope heterogeneity is questionable for valid estimation and inference. 
To accommodate heterogeneity among individuals, we relax the common slope assumption by allowing  $\theta_{g0}$ to follow a cluster pattern of a general form
\begin{equation*}
\theta_{g0}=\sum_{k\in[K_0]}1\{g\in C_{k0} \}\vartheta_{k0} \text{ with }\vartheta_{k0}=(\overline{\beta}_{k0},\beta_{k0}')',
\end{equation*}
where the latent clusters $C_{k0}$, $k\in[K_0]$, consist of a partition of groups. The number of latent clusters, $K_0$, is finite.  Let $G_k=\#C_{k0}$  be the number of groups in each cluster $k$; we posit that $G_k/G\to c_k\in (0,1)$ as $G\to \infty$, i.e., all latent clusters are non-vanishing asymptotically. In practice, when $K_0$ is unknown, we will provide an information criterion (IC) in Section \ref{sec:post-classification} to determine it jointly with the cluster memberships. The latent structure here is different from the mixture data assumption considered in \citet{kasahara2009nonparametric}, which assumes that choice probability is a finite mixture of base distributions. 

A modified Penalized Profile Likelihood estimator is adapted from the NPL with fixed effects. Let the fixed effects estimator be 
$$
\widehat\mu_g(\theta_g,P_g)= \underset{\mu_g\in\mathcal{A}}{\arg\min}\ \Psi_{ng}(\mu_g,\theta_g;P_g).
$$
For each group $g$, define the profile negative log-likelihood
\begin{equation*}
Q_{ng}(\theta_g;P_g)=\frac{1}{n}\sum_{i\in[n]}\psi(w_{ig},P_g;\theta_g,\widehat\mu_g(\theta_g,P_g)).
\end{equation*} 
Given $G$ groups, the average negative profile log-likelihood of all groups is
\begin{equation*}
Q_{n}(\bm\theta;\bm{P})=\frac{1}{G}\sum_{g\in[G]}Q_{ng}(\theta_g;P_g).
\end{equation*}
Let $P_{ig}(\theta_g)$ be the fixed points that solve
$$
P_{ig}(\theta_g)=\Lambda(\overline{P}_{ig}\overline\beta_g+X_{ig}'\beta_g+\widehat{\mu}_g(\theta_g,P_g)),\ i\in[n],g\in[G],
$$
and denote $\bm{P}(\bm\theta)=\{P_{ig}(\theta_g)\}_{i\in[n],g\in[G]}$.

To sort the groups into clusters, we employ the C-Lasso method as a classifier. Define the NPL-Classifier-Lasso (NPL-C-Lasso) estimator as
\begin{equation}\label{eq:C-lasso_criterion}
    (\widetilde{\bm\theta},\widetilde{\bm\vartheta})= \underset{(\bm\theta,\vartheta)\in \Theta^{G+K_0}}{\arg\min}~Q_{n}(\bm\theta;\bm{P}(\bm\theta))+\frac{\rho}{G}\sum_{g\in[G]}\prod_{k\in[K_0]}\|\theta_g-\vartheta_k\|_2,
\end{equation}
where $\rho$ is a tuning parameter to determine the level of the multiplicative classification penalty. 
It encourages clustering among the group-specific estimates $\{\widetilde{\theta}_g\}_{g\in[G]}$ by shrinking $\theta_g$ toward the nearest cluster estimates $\{\widetilde\vartheta_k\}_{k\in[K_0]}$. The multiplicative structure $\prod_{k\in[K_0]} \|\theta_g-\vartheta_k\|_2$ ensures that $\theta_g$ is penalized more heavily when it is far from \emph{all} cluster estimates, effectively promoting exact equality $\widetilde{\theta}_g = \widetilde\vartheta_k$ for some $k$. 
The optimization problem yields the estimates $\{\widetilde{\theta}_g,\widetilde{\vartheta}_k\}$ for schools $g\in[G]$ and groups $k\in[K_0]$. A group $g$ is classified  into the cluster $\widetilde{C}_k$  if $\widetilde\theta_g$ is closest to $\widetilde{\vartheta}_k$ in $L_2$-distance, that is 
$$
\widetilde{C}_k=\{g\in[G]:\underset{l\in[K_0]}{\arg\min}\ \| \widetilde{\theta}_g-\widetilde{\vartheta}_l\|_2=k\}
$$ 
for $k\in[K_0]$. Asymptotically, this rule is equivalent to the equality rule $\widetilde{C}^e_k=\{g\in[G]:\widetilde{\theta}_g=\widetilde{\vartheta}_k\}$ as in K-means.\footnote{
In principle, any consistent classification method can determine the latent structure, for example, K-means \citep{bonhomme2015grouped}.
Numerically, the K-means algorithm relies heavily on the estimation quality of $\widehat\theta_g$s, and is thus computationally heavy. 
} In finite samples, C-Lasso is more lenient in its restrictions.

\subsection{Adjusted Algorithm}\label{sec:adjusted-algorithm}
The C-Lasso minimization problem  \eqref{eq:C-lasso_criterion} is non-convex and it requires estimating CCPs $\bm{P}(\bm\theta)$ in each evaluation, which is computationally demanding.
While the adjusted algorithm does not exactly solve \eqref{eq:C-lasso_criterion}, the resulting estimates are close in practice because $\widetilde{\bm{P}}$ is a consistent proxy for the unknown CCPs. Directly solving \eqref{eq:C-lasso_criterion} would require re-computing the CCP fixed point $\bm{P}(\bm\theta)$ at every evaluation of the C-Lasso objective, increasing computation time by orders of magnitude.
To make the estimation computationally tractable, we devise a new algorithm to decouple the classification step from the estimation process to enhance computational stability. It consists of three steps, as illustrated in the diagram Figure \ref{fig:NPL-C-lasso-estimator}: 

\medskip\medskip\medskip
\noindent\underline{1. First Step NPL:} apply the NPL method (Algorithm \ref{alg:NPL_alg_FE}) to each group $g$, obtaining estimators for the CCPs $P_g$, denoted $\widetilde{P}_g$, $g\in[G]$. 

\medskip\medskip
\noindent\underline{2. C-Lasso:} estimate the latent clusters via solving C-Lasso:
\begin{equation*}
     (\widetilde{\bm\theta},\widetilde{\bm\vartheta})= \underset{(\bm\theta,\vartheta)\in \Theta^{G+K_0}}{\arg\min}~Q_{n}(\bm\theta;\widetilde{\bm{P}})+\frac{\rho}{G}\sum_{g\in[G]}\prod_{k\in[K_0]}\|\theta_g-\vartheta_k\|_2,
\end{equation*}
where $\widetilde{\bm{P}}=\{\widetilde{P}_g\}_{g\in[G]}$ is independent of $\bm\theta$.
This allows us to classify groups into $K_0$ distinct clusters, $\widetilde{C}_1,\dots,\widetilde{C}_{K_0}$, \emph{without} sequentially updating CCPs. The number of latent clusters $\widehat{K}$ is determined by an IC, see Section \ref{sec:post-classification}. 

\medskip\medskip
\noindent\underline{3. Post-Classification Estimation:} re-estimate $\theta_k$ for each cluster $k$ through NPL with FE (Algorithm \ref{alg:NPL_alg_FE}) by pooling all groups' observations within the cluster. Post-classification estimators are denoted as $\widehat{\theta}_{\widetilde{C}_k}$, $k\in[K_0]$. 
Apply Algorithm \ref{alg:Bootstrap} to groups within each cluster to obtain the confidence intervals for common parameters.
\footnote{Step~3 jointly re-estimates the CCPs alongside the parameters, so $(\widehat{\theta}_{\widetilde{C}_k},\widehat{P}_{\widetilde{C}_k})$ satisfy the NPL fixed point at the cluster level. The original group-level fixed points from Step~1 need not be preserved, as the first-step CCPs serve only as inputs to C-Lasso and play no further role in post-classification inference.} 
\medskip\medskip\medskip

We will show that the C-Lasso step classifies groups into their true clusters with probability approaching one (w.p.a.1.)~in Section \ref{sec:classification}. Therefore, inference based on post-classification estimators is asymptotically valid, similar to the homogeneous model. The empirical application concerning risky behaviors among students reveals that, if we ignored the latent structure, we would severely underestimate the peer effects for the cluster with significant effects and draw erroneous conclusions about the intelligence variables. Therefore, we recommend the above algorithm to cope with possible slope heterogeneity in real-world applications.

\begin{figure}
\centering
\includegraphics[width=0.75\linewidth]{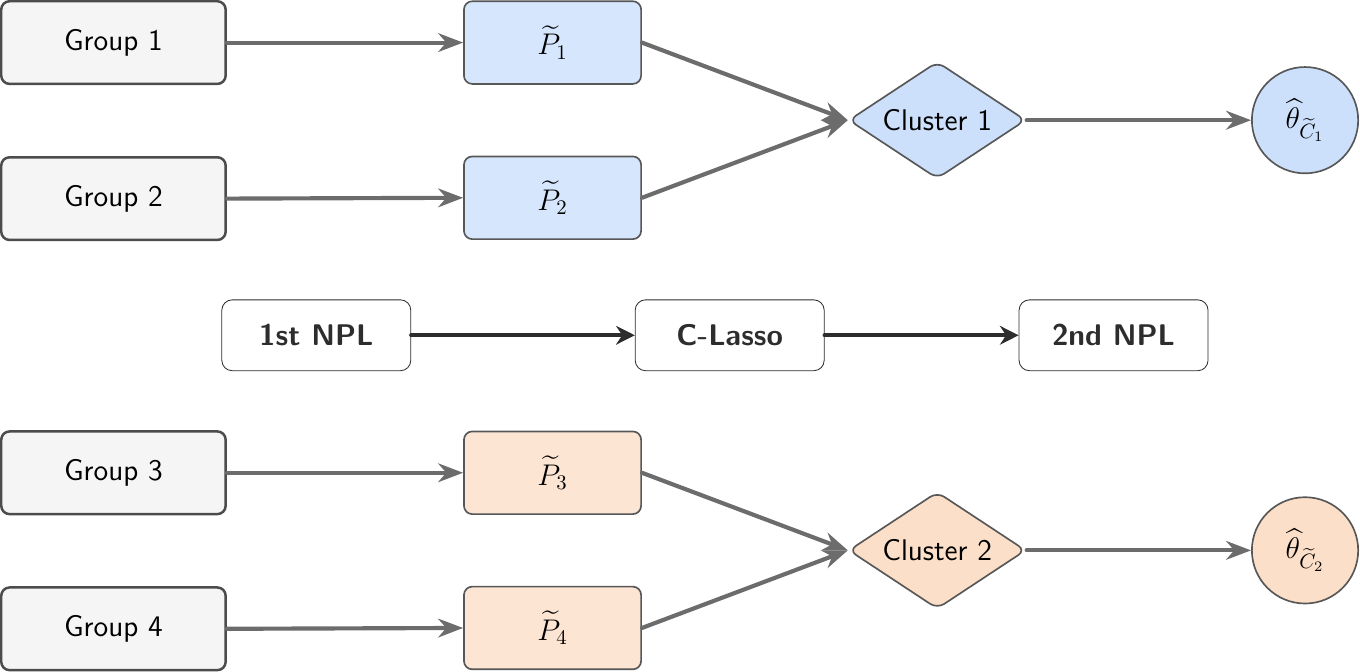}
\caption{Diagram of classification and estimation procedure.
It displays an example of the adjusted algorithm with four groups classified into two clusters. 
}\label{fig:NPL-C-lasso-estimator}
\end{figure}

\subsection{Average Partial Effects}\label{sec:APE}

Parameters $\beta_g$ in \eqref{eq:model} operate on the
latent utility scale and are embedded in the nonlinear functional forms. The
APEs are more straightforward and intuitive in connecting the covariates with the
choice probabilities.

The derivative of $P_{ig}$ with respect to the $p$th covariate $X_{ig,p}$, supposing it is a continuous variable, satisfies
\begin{equation}
\label{eq:APE-equation}
\frac{d P_{ig}}{dX_{ig,p}}=P_{ig}(1-P_{ig})\left(\overline\beta_g\frac{d\overline P_{ig}}{dX_{ig,p}} +\beta_{g,p}\right)\text{ for }i\in[n],g\in[G],
\end{equation}
where we use the definition of equilibrium CCP \eqref{eq:CCP-def} and the chain rule of differentiation.
Let
$D_{g} = \operatorname{diag}\big(\{P_{ig}\}_{i\in[n]}\big)$
and let $W^g$ be the row-normalized adjacency matrix with $W^g_{ij}=1\{N_i^g>0\}\cdot F^g_{ij}/N_i^g$.
Solving \eqref{eq:APE-equation} and taking the average across $n$ units,
the APE of the $p$th covariate in the group $g$ is
\begin{equation}\label{eq:ape}
  \delta_{g,p} =
  \frac{1}{n}\mathbf{1}_n'
  \big(I_n - \overline{\beta}_{g} D_gW^g\big)^{-1}
  D_{g} \beta_{g,p}\mathbf{1}_n,
\end{equation}
where $\mathbf{1}_n=(1,\dots,1)'\in\mathbb{R}^n$.
The population APE for the $k$th cluster under the true parameter values is
\begin{equation}
    \label{eq:cluster-APE}
    \delta^0_{k,p} = \frac{1}{G_k}\sum_{g\in C_{k0}} \delta^0_{g,p}.
\end{equation}
Throughout, $\delta^0_{g,p}$ and $\delta^0_{k,p}$ denote the population group-$g$ and cluster-$k$ APEs, i.e., the functional in~\eqref{eq:ape} evaluated at the true parameters $(\theta_{g0}, P_{g0})$ and $(\theta_{k0}, P_{k0})$ respectively; estimators carry hats.
A summary measure of the overall average partial effect across all clusters is the population-level weighted APE
\begin{equation}
    \label{eq:weighted-APE}
    \overline{\delta}^0_p = \sum_{k\in[K_0]} \frac{G_{k}}{G} \delta^0_{k,p},
\end{equation}
where the weights are the cluster proportions $G_{k}/G$.
We do not define the APE for the peer effect due to its interactive nature. From \eqref{eq:ape}, the peer effect behaves as a social multiplier, which enlarges the covariates' influence when $\overline\beta_g>0$. 
For covariates with discrete supports, the APE is the average of the differences of the equilibrium CCPs at two different levels. In this case, the APEs lack closed-form representations. 

\medskip

The algorithms we have laid out so far provide a practical way for economists to process their data to evaluate the peer effects by accommodating rich heterogeneity across groups as well as latent structures over clusters. 
We stack multiple modules of the algorithms and tailor the implementation to speed up the computation, making a pipeline that is ready for complex data collected via large-scale surveys or administrative records.
Next, we establish the large sample theory to support the procedures we have proposed.

\section{Large Sample Theory}\label{sec:large_sample_results}

In our asymptotic analysis, the group size $n$ is sent to infinity. The number of groups can stay finite, or diverge to infinity.
For the latter case, we adopt a triangular limit framework, viewing $G$ as a deterministic function of $n$.
In asymptotic statements, we explicitly send $n\to \infty$, understanding that $G=G(n)$ and $G(n)/n\to c\in (0,\infty)$ as $n\to \infty$.

We first set up several basic assumptions. 
Given these assumptions,  we analyze in Section \ref{sec:firstNPL} the NPL estimators $\widetilde\zeta_g=(\widetilde{\mu}_g,\widetilde{\theta}_g)$ and CCPs $\widetilde{P}_g$ for each group $g$, which are the building blocks of our adjusted algorithm. Section \ref{sec:asymptotic_normality} establishes asymptotic normality of the NPL fixed effects estimator $\widehat\theta$ and the parametric bootstrap debiased estimator $\widehat\theta^*$. We present the classification consistency and an IC for selecting the number of clusters in Sections \ref{sec:classification} and \ref{sec:post-classification}, respectively. Let $C<\infty$ denote a finite positive constant.
\begin{assumption}
\label{npl}
\begin{enumerate}[label=(\roman*),font=\upshape]
\item \label{npl-i}
   The covariate $X_{ig}$ satisfies $\max_{i,g}\| X_{ig}\|_\infty<C$.
\item\label{npl-ii} There exists open sets $\mathcal{A}_0\subset\mathcal{A}$ and $\Theta_0\subset\Theta$, such that the true values $(\mu_{g0},\theta_{g0}')'$ are interior points of $\mathcal{A}_0\times\Theta_0$ for all $g$. There exists a constant $\underline{c}\in(0,1/2)$ such that the true CCPs $P_{g0}$'s lie in the interior of $(\underline{c},1-\underline{c})^n$ for all $g$. 
\item\label{npl-iii} The peer effects satisfy $0<\max_{k\in[K_0]}|\overline{\beta}_{k0}|<4$.
\item\label{npl-iv} $(\zeta_{g0},P_{g0})$ is an isolated population NPL fixed point to \eqref{eq:NPL_fixed_point_population} for each $g\in[G]$.\footnote{An isolated fixed point means there exists a neighborhood of $(\zeta_{g0}, P_{g0})$ containing no other NPL fixed point. Theoretically, uniqueness is generically satisfied under the moderate social interactions condition in Assumption~\ref{npl}\ref{npl-iii} and leads to the isolation.} 
\item\label{npl-v} $\max_{j\in[n]}\sum_{i\in[n]}(N_i^g)^{-1}F^g_{ij}<C$ for each $g\in[G]$.
\end{enumerate}
\end{assumption}

Assumption \ref{npl}\ref{npl-i} and \ref{npl-ii} guarantee a finite single index, and hence $P_{ig0}\in(\underline{c},1-\underline{c})$ for some finite constant $\underline{c}\in(0,1/2)$.
Although technically we can relax \ref{npl-i} to allow light-tailed covariates, such an extension complicates notation but brings no theoretical insights. Part \ref{npl-iii} of Assumption \ref{npl} is called the ``moderate social interactions'' condition in the literature \citep[see][]{glaeser2003nonmarket,horst2006interactions,xu2018social,lin2017estimation}. This condition is sufficient for the existence and uniqueness of the Bayesian Nash equilibrium. Similar conditions to Part \ref{npl-iv} of Assumption \ref{npl} are imposed in \cite*{lin2021uncovering} and \citet{lin2024binary}. This condition ensures that the NPL estimation of unknown parameters and CCPs is well-defined and consistent for all groups.  Note
$$
\sum_{i\in[n]}(N_i^g)^{-1}F^g_{ij}\leq \sum_{i\in[n]}F^g_{ij}=N_j^g,
$$
from which a sufficient condition for part \ref{npl-v} is that students in each group have only a finite number of friends, which usually holds in practice. This condition ensures that the estimation errors of $\widetilde{P}_{ig}$'s do not significantly affect the criterion function through social interactions.
Therefore, we can replace $P_0$ with the estimate $\widetilde{P}$  in the C-Lasso step.

\subsection{First-step NPL Estimator}\label{sec:firstNPL}
Define $\zeta_g=(\mu_g,\theta_g)$ and $\zeta_{g0}=(\mu_{g0},\theta_{g0})$. By applying NPL to each group $g$, we obtain the estimators $\widetilde\zeta_g$ and $\widetilde{P}_g$.  Our first theoretical result, Proposition \ref{prop:first_step_npl}, is an extension of Proposition 2 in \citet{aguirregabiria2007sequential} and Theorem 1 in \citet{lin2024binary}. We focus on the consistency of $\widetilde\zeta_g$ and the convergence rate of $\widetilde{P}_g$. 

\begin{proposition}
\label{prop:first_step_npl}Suppose Assumption \ref{npl} and Assumption \ref{assu:invertibility_assumption1} in Appendix \ref{appendix:first-NPL} hold.\footnote{
To study the global consistency of $\widetilde{P}_g$ over $g\in[G]$,  Assumption \ref{assu:invertibility_assumption1}  is needed to guarantee the invertibility of the Jacobian matrices with respect to the unknown parameters. However, due to its technical nature, we will present it in Appendix \ref{appendix:first-NPL}.
} Then 
\begin{enumerate}[label=(\alph*),font=\upshape]
    \item\label{first-step-consistency} $\max_{g\in[G]}\|\widetilde\zeta_g-\zeta_{g0}\|=o_p(1)$;
    \item\label{first-step-probability-vector} there exists a positive constant $C$ such that the first-step NPL estimator for $\widetilde{P}$ satisfies
    $$\max_{g \in [G]}\left\|\widetilde{P}_g-P_{g0}\right\|_2\leq C\sqrt{\log(Gn)}$$
with probability at least $1-2G^{-1}n^{-2}$. 
\end{enumerate} 
\end{proposition}
Proposition \ref{prop:first_step_npl}\ref{first-step-consistency} shows that $\widetilde{\zeta}_g$'s are consistent uniformly for $g\in[G]$. In the steps of C-Lasso or second-step NPL, we can actually start from these consistent first-step estimators in the optimization. By using this optimization strategy, it is also guaranteed that the post-classification estimators (and also $\widehat\theta$) are consistent. Therefore, we can focus on the asymptotic normality in the next subsection.
Part \ref{first-step-probability-vector} is a non-asymptotic result, derived from Hoeffding's inequality \citep{hoeffding1963probability}.  Intuitively, the NPL estimator $\widetilde{\theta}_g$  is expected to be $\sqrt{n}$-consistent as each school $g$ has $n$ students. Then by Assumption \ref{npl}, $\widetilde{P}_g$ is a continuous function of $\widetilde{\theta}_g$ around $\theta_{g0}$, and hence this $\sqrt{n}$-convergence rate of $\widetilde{P}_g$ is reasonable.
This result allows directly applying C-Lasso with $P_{g0}$  replaced by $\widetilde{P}_g$; as we will see, the classification consistency only requires  $\sqrt{n}$-consistent estimators for the CCPs uniformly over $G$ groups. Moreover, $P_g$  is a function of school-specific fixed effects $\mu_g$, which means the convergence rate cannot outperform $n^{-1/2}$ in general.

\subsection{Inference for Model with Homogeneous Slopes}\label{sec:asymptotic_normality}

We establish asymptotic normality of $\widehat\theta$ and $\widehat\theta^*$ in this subsection. Inference based on the post-classification estimators is essentially the same. The following theorem justifies statistical inference based on $\widehat\theta$ and the parametric bootstrap estimator $\widehat\theta^*$.
Let $H_{0}$ be a concentrated Hessian matrix, $B_{0}$ be the asymptotic bias term, and $\Omega_{0}$ denote the asymptotic covariance matrix, whose explicit forms are defined by Equations \eqref{Hk0}, \eqref{Bk0}, and \eqref{omegak0}, respectively, in Appendix \ref{appendix:post-estimation}.

\begin{theorem}\label{thm:post_classification_estimation}
If conditions in Proposition \ref{prop:first_step_npl} and Assumption \ref{assu:invertibility2} in Appendix \ref{appendix:post-estimation} hold,\footnote{
Assumption \ref{assu:invertibility2} makes sure that some necessary Jacobian matrices are invertible. Due to the technicality, these conditions are presented in Appendix \ref{appendix:post-estimation}.
} then
\begin{enumerate}[label=(\alph*),font=\upshape]
\item\label{thm:post_classification_estimation-a}
        $\sqrt{nG}\Big(\widehat{\theta}-\theta_{0}\Big)-H_{0}^{-1}B_{0}\stackrel{d}{\to} N(0,H_{0}^{-1}\Omega_{0}H_{0}^{-1});$
\item\label{thm:post_classification_estimation-b}
        $\sqrt{nG}\Big(\widehat{\theta}^*-\theta_{0}\Big)\stackrel{d}{\to}N(0,H_{0}^{-1}\Omega_{0}H_{0}^{-1})$;
\item\label{thm:post_classification_estimation-c} 
for any non-random vector $a\in\mathbb{R}^{1+\mathrm{dim}(X_{ig})}$ with $\|a\| = 1$, we have
$$\sup_{x\in\mathbb{R}}\left|\Pr\left(\sqrt{nG}a'(\widehat\theta^b-\widehat\theta)\leq x \;\Big|\; \boldsymbol W\right)-\Pr\left(\sqrt{nG}a'(\widehat\theta-\theta_0)\leq x\right)\right|=o_p(1),$$ 
where $\boldsymbol W= \{w_{ig}\}_{i\in[n],g\in[G]}$ is all the available data on which the bootstrap is conditioned.
\end{enumerate}
\end{theorem}

Theorem \ref{thm:post_classification_estimation}, our main theoretical contribution,
consists of three parts. Part \ref{thm:post_classification_estimation-a} shows that $\widehat{\theta}$ is asymptotically normal, but it exhibits an incidental parameter bias $H_{0}^{-1}B_{0}$ that shifts its center away from the true parameter $\theta_{0}$. 
Though the incidental parameter issue is known in the panel data literature,
our model involves games through fixed points, and no prior work has explored the presence and the
forms of such a bias due to the complexity that the NPL algorithm affects how the estimation error from $\widehat\mu_g$ accumulates.


Once we diagnose the bias as arising from the high-dimensional incidental parameters, we borrow from existing methods for a remedy.
We verify in our context that the parametric bootstrap inference method delivers
an asymptotically unbiased $\widehat{\theta}^*$. Part \ref{thm:post_classification_estimation-b} justifies the unbiasedness of $\widehat\theta^*$ and the correct coverage rates of the bootstrap confidence intervals, and Part \ref{thm:post_classification_estimation-c}  reveals the underlying mechanism.
Conditioning on the realized data $\boldsymbol{W}$, the ``true parameter'' in the world of the parametric bootstrap is $\widehat{\theta}$. The asymptotic distribution of $\widehat\theta^b-\widehat\theta$ mimics that of $\widehat\theta-\theta_0$, sharing the same asymptotic bias. As a result, 
$$
\sqrt{nG}\Big(\widehat{\theta}^b-\theta_{0}\Big) = \sqrt{nG}\Big( (\widehat{\theta}^b - \widehat\theta )  - (\widehat\theta -\theta_{0})\Big) 
$$
is properly centered at 0 by canceling out the common bias.

\subsection{Classification Consistency}\label{sec:classification}
In this subsection, we establish the classification consistency of C-Lasso when there is a latent structure.
To evaluate the classification accuracy, we define two events $$\widetilde{E}_{knG,g}=\{g\notin \widetilde{C}_k|g\in C_{k0}\}\qquad\text{and}\qquad \widetilde{F}_{knG,g}=\{g\notin C_{k0}|g\in \widetilde{C}_k\},$$
which are analogous to Type I and Type II errors in hypothesis testing, respectively.
The false exclusion $\widetilde{E}_{knG,g}$ indicates that a group $g$ with its true cluster identity $C_{k0}$
is assigned to a wrong cluster, and the false inclusion $\widetilde{F}_{knG,g}$ implies that a group $g$ is mistakenly assigned to $\widetilde{C}_k$.
Let 
$$\widetilde{E}_{knG}=\cup_{g\in C_{k0}}\widetilde{E}_{knG,g} \quad \mbox{(any false exclusion from cluster $k$)}$$
and 
$$\widetilde{F}_{knG}=\cup_{g\in \widetilde{C}_k}\widetilde{F}_{knG,g}
\quad \mbox{(any false inclusion into cluster $k$)}, $$
Theorem \ref{thm:classification_consistency} establishes that, when the number of clusters is correctly specified,  C-Lasso ensures asymptotically perfect classification.

\begin{theorem}\label{thm:classification_consistency}
Under the assumptions in Proposition \ref{prop:first_step_npl}, we have
\begin{align*}
  \Pr\Big(\cup_{k\in[K_0]}\widetilde{E}_{knG}\Big) & \leq \sum_{k\in[K_0]}\Pr\Big(\widetilde{E}_{knG}\Big)\to 0 \text{ and} \\
  \Pr\Big(\cup_{k\in[K_0]}\widetilde{F}_{knG}\Big) & \leq \sum_{k\in[K_0]}\Pr\Big(\widetilde{F}_{knG}\Big)\to 0.
\end{align*}
\end{theorem}
Theorem \ref{thm:classification_consistency} demonstrates that all schools are classified into their true clusters w.p.a.1.~if we know $K_0$ \textit{a priori}. In the next subsection, we will propose an information criterion (IC) to select the number of clusters in practice, building on the post-classification estimators.

\subsection{Information Criterion}\label{sec:post-classification}
We re-estimate the unknown parameters via NPL using within-cluster school data, denoting the resulting post-classification estimators as $$\widehat{\theta}_{\widetilde{C}_k} =(\widehat{\overline\beta}_{\widetilde{C}_k},\widehat{\beta}_{\widetilde{C}_k}')' \quad \mbox{for } k\in[K].$$ As we mentioned before,  $\widehat{\theta}_{\widetilde{C}_k}$ is useful for constructing the IC and determining $\widehat{K}$. Within each cluster, we can apply the parametric bootstrap to conduct statistical inference based on Theorem~\ref{thm:post_classification_estimation}.

While our preceding analysis presumes prior knowledge of $K_0$, the true number of clusters is generally unknown in empirical applications. To address this practical consideration, we posit that $K_0$ is bounded above by a finite integer $K_{\max}$ and develop a data-driven selection procedure via IC.  We have the C-Lasso estimates $\{\widetilde{\theta}_g(K),\widetilde{\vartheta}_k(K)\}$ of $\{\theta_{g0},\vartheta_{k0}\}$, where we make the dependence of $\widetilde{\theta}_g$ and $\widetilde{\vartheta}_k$  on $K$  explicit. As above, we classify a school $g$ into cluster $\widetilde{C}_k(K)$ if and only if $\widetilde{\theta}_g(K)$ is closest to one in   $\{\widetilde{\vartheta}_k(K)\}_{k \in [K_{\max}]}$.  We propose to select $K$ by minimizing the IC
$$
\mathrm{IC}(K)= \frac{1}{nG}\sum_{k\in[K]}\sum_{g\in\widetilde{C}_k(K)}\sum_{i\in[n]}\psi\Big(w_{ig},\widehat{P}_g;\widehat{\theta}_{\widetilde{C}_k(K)},\widehat{\mu}_g\big(\widehat{\theta}_{\widetilde{C}_k(K)},\widehat{P}_g\big)\Big)+\lambda pK,
$$
where $\lambda$  is a tuning parameter and we recommend setting $\lambda=\frac{1}{4}\log(\log n)/n$, in accordance with \citet*[Section 2.5]{su2016identifying}. Let $$\widehat{K}=\arg\min_{K\in[K_{\max}]}\mathrm{IC}(K),$$ where $K_{\max}\geq K_0$ is a pre-specified integer. The following proposition justifies the use of the IC.
\begin{proposition}\label{prop:K-selection}
    Suppose conditions in Proposition \ref{prop:first_step_npl} hold, if $K_{\max}\geq K_0$ and $\lambda + (n\lambda)^{-1}\to 0$ as $n\to\infty$, then
$\Pr\big(\widehat{K}=K_0\big) \to 1$.
\end{proposition}
Proposition \ref{prop:K-selection} shows that the true number of clusters can be consistently estimated, thus closing the loop of estimation and inference.
In Figure \ref{fig:summary-procedure}, we recap the building blocks of estimation, model selection, and inference when $K_0$ is unknown. This scheme makes the implementation feasible with real data.
\afterpage{
\begin{figure}[!htbp]
\centering
\resizebox{1\linewidth}{!}{%
\begin{tikzpicture}[
    stage/.style={
        rectangle,
        draw,
        rounded corners=3pt,
        align=center,
        minimum width=5.8cm,
        minimum height=2.4cm,
        font=\small
    },
    arrow/.style={
        ->,
        thick,
        rounded corners,
        >=Stealth
    }
]

\node (stage1) [stage] {
\textbf{Estimation} \\[0.4em]
For $K = 1,\dots,K_{\max}$, \\ apply the adjusted algorithm in Section \ref{sec:adjusted-algorithm},\\
obtain $\{\widehat{\theta}_{\widetilde{C}_k}\}_{k\in[K]}$
};

\node (stage2) [stage, right=1.8cm of stage1] {
\textbf{Model Selection} \\[0.4em]
Compute $\text{IC}(K)$, \\[0.2em]
select $\widehat{K},$\\
retrieve $\{\widehat{\theta}_{\widetilde{C}_k}\}_{k\in[\widehat{K}]}$ as the point estimators
};

\node (stage3) [stage, below=2.0cm of $(stage1)!0.5!(stage2)$] {
\textbf{Inference} \\[0.4em]
Apply Algorithm~\ref{alg:Bootstrap} to each cluster, \\[0.2em]
obtain bias-corrected estimators $\{\widehat{\theta}^*_{\widetilde{C}_k}\}_{k\in[\widehat{K}]}$\\[0.2em]
and their confidence intervals
};

\draw [arrow] (stage1.east) -- (stage2.west);

\draw [arrow]
    (stage1.south) 
    .. controls +(0,-1.0) and +(-1.0,0) ..
    (stage3.west);

\draw [arrow]
    (stage2.south) 
    .. controls +(0,-1.0) and +(1.0,0) ..
    (stage3.east);

\end{tikzpicture}%
}
\caption{Implementation scheme when $K_0$ is unknown}\label{fig:summary-procedure}
\end{figure}
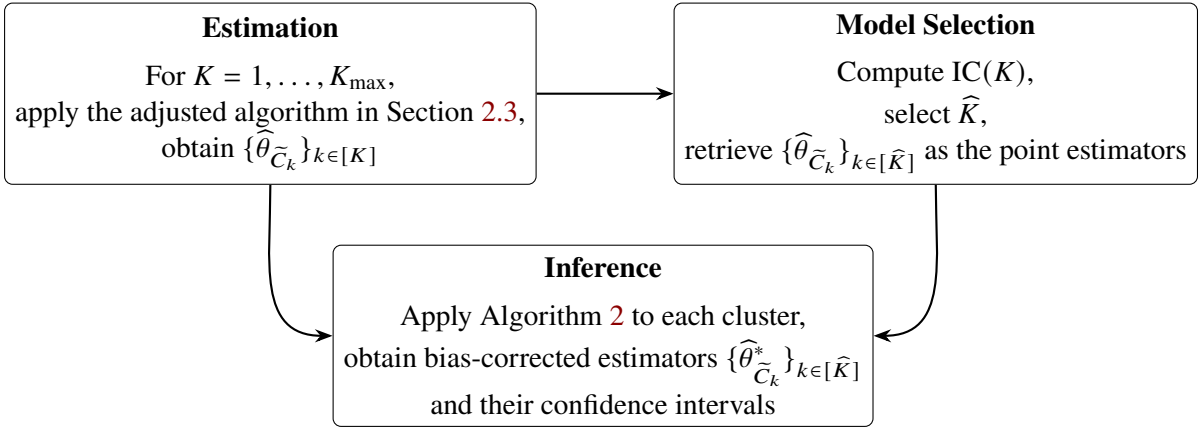
}
\subsection{Inference on Average Partial Effects}\label{sec:ape_inference}

We now establish the asymptotic distribution of the debiased APE estimator.
Recall from~\eqref{eq:cluster-APE} that the cluster-$k$ APE is a smooth functional $\delta^0_{k,p} = \delta(\theta_{k0}, P_{k0})$.
Let $\nabla\delta^0_{k,p} = \partial\delta^0_{k,p}/\partial\theta_k'$ be its gradient.
From~\eqref{eq:ape}, the gradient has the closed form
\begin{equation}\label{eq:ape_gradient}
\nabla\delta^0_{k,p}
= \frac{1}{G_k}\sum_{g\in C_{k0}}
\frac{1}{n}\mathbf{1}_n'
\big(I_n - \overline{\beta}_k D_g W^g\big)^{-1}
\Big[
  D_g W^g \big(I_n - \overline{\beta}_k D_g W^g\big)^{-1} D_g \beta_k \mathbf{1}_n
  \,{:}\,
  D_g \mathbf{1}_n
\Big],
\end{equation}
where the two blocks correspond to the partial derivatives with respect to $\overline{\beta}_k$ and $\beta_k$, respectively.
Let $\widehat{\delta}_{k,p}^{\,*}$ be the debiased APE estimator obtained by evaluating the functional~\eqref{eq:cluster-APE} at $(\widehat{\theta}^*_{\widetilde{C}_k}, \widehat{P}_{\widetilde{C}_k})$, where $\widehat{\theta}^*_{\widetilde{C}_k}$ is the post-classification debiased estimator and $\widehat{P}_{\widetilde{C}_k}$ is the within-cluster CCP estimator.

\begin{corollary}\label{cor:ape_asymptotics}
Under the conditions of Theorem~\ref{thm:post_classification_estimation} and \ref{thm:classification_consistency}, with the gradient $\nabla\delta^0_{k,p}$ given by~\eqref{eq:ape_gradient} and cluster-specific matrices $H_{0k}$, $\Omega_{0k}$ defined as the analogs of $H_0$, $\Omega_0$ in Appendix~\ref{appendix:post-estimation}:
\begin{enumerate}[label=(\alph*),font=\upshape]
\item\label{cor:ape_asymptotics-a} $\sqrt{nG_k}\,\big(\widehat{\delta}_{k,p}^{\,*} - \delta^0_{k,p}\big) \stackrel{d}{\to}
    N\big(0,\; (\nabla\delta^0_{k,p})' \,H_{0k}^{-1}\Omega_{0k}H_{0k}^{-1}\, \nabla\delta^0_{k,p}\big)$;
\item\label{cor:ape_asymptotics-b}
conditioning on the data $\boldsymbol W$, the $b$th bootstrap estimator $\widehat{\delta}_{k,p}^{\,b}$ for the APE follows
    $$\displaystyle\sup_{a\in\mathbb{R}}\Big|
        \Pr\Big(\sqrt{nG_k}\big(\widehat{\delta}_{k,p}^{\,b} - \widehat{\delta}_{k,p}\big)\leq x \;\Big|\; \boldsymbol W\Big)
      - \Pr\Big(\sqrt{nG_k}\big(\widehat{\delta}_{k,p} - \delta^0_{k,p}\big)\leq x\Big)
    \Big| = o_p(1).$$
\end{enumerate}
\end{corollary}

Part~\ref{cor:ape_asymptotics-a} is implied by applying the delta method to the consistent estimator  $\widehat{P}_{\widetilde{C}_k}$ in Theorem~\ref{thm:post_classification_estimation}\ref{thm:post_classification_estimation-a}.
Theorem~\ref{thm:classification_consistency} ensures that replacing $C_{k0}$ with the estimated partition $\widetilde{C}_k$ does not affect the asymptotic distribution.
The gradient~\eqref{eq:ape_gradient} reveals the social multiplier channel: as $\overline{\beta}_k$ grows, the factor $(I_n - \overline{\beta}_k D_g W^g)^{-1}$ amplifies both the APE and its sampling variance.

The population-level weighted APE defined in~\eqref{eq:weighted-APE} is a linear combination of the cluster-level APEs.
Let $\widehat{\overline{\delta}}_{p}^{\,*}=\sum_{k=1}^{K} (\widehat{G}_k/G)\,\widehat{\delta}_{k,p}^{\,*}$, where $\widehat{G}_k=\#\widetilde{C}_k$ is the estimated cluster size.
Since $\widehat{G}_k/G\stackrel{p}{\to}G_{k}/G$ (Theorem~\ref{thm:classification_consistency}) and each $\widehat{\delta}_{k,p}^{\,*}$ is asymptotically normal (Corollary~\ref{cor:ape_asymptotics}), the continuous mapping theorem implies that $\widehat{\overline{\delta}}_{p}^{\,*}$ is also asymptotically normal with variance
\(\sum_{k=1}^{K} (G_{k}/G)^2\,\mathrm{Var}(\widehat{\delta}_{k,p}^{\,*})\).
The parametric bootstrap procedure extends directly to the weighted APE by computing the weighted averages across bootstrap replications, yielding valid confidence intervals for $\overline{\delta}^0_p$. For discrete covariates, the APE is defined as a finite difference of CCPs. The delta method applies analogously since the APE remains a smooth functional of $(\theta_k, P_k)$.

\section{Monte Carlo Simulations} \label{sec:MC}
We check the finite sample performance of our estimation and inference methods in Monte Carlo experiments. For each individual $i$, we first randomly select the number of friends, $N_i^g$, from a uniform distribution over $\{0\}\cup[N_{\max}]$ with $N_{\max}=5$. Then we randomly form $N_i^g$ friendship links with other individuals in the group $g$ to construct the friend matrix $F^g$.
We consider the following DGP with network interactions through $F^g$:
\begin{equation}
Y_{ig}=1\left\{\overline{P}_{ig0}\overline\beta_{g0}+X_{ig}'\beta_{g0}+c_{g0}+\mu_{g0}>\varepsilon_{ig}\right\},\quad i\in[n],g\in[G],
\end{equation}
where $c_{g0}+\mu_{g0}$ is a composite fixed effect, with $\mu_{g0}$ being standard normal and $c_{g0}$ taking cluster-specific values $c_{k0}$. The idiosyncratic error $\varepsilon_{ig}$ is standard logistic, independent across $g$ and $i$, and independent of all regressors.  The CCPs, $P_{ig0}$, are solved using a fixed point algorithm on
\begin{equation}
P_{ig0}=\Lambda \left(\overline{P}_{ig0}\overline\beta_{g0}+X_{ig}'\beta_{g0}+c_{g0}+\mu_{g0} \right),\quad i\in[n],g\in[G].
\end{equation}
The observations in each DGP are drawn from two clusters with the proportion $G_1:G_2 =
0.6:0.4.$ The exogenous regressor $X_{ig}=0.1\mu_{g0}+e_{ig}$ where $e_{ig}\sim i.i.d.~ N(0,1)$. We set the number of groups $G \in  \{100, 200\}$, and the group size $n \in \{50,100,200\}$. We set the true coefficients $(\overline\beta_{k0},\beta_{k0},c_{k0})$ as $(1.5,-1,-0.5)$ and $(0,1,0)$, respectively, for the two clusters.
The number of Monte Carlo replications is 1000 for each DGP, and 500 parametric bootstrap replications are used for the bias correction. 

\begin{table}[!ht]
\centering
\caption{Percentage of selecting $K=1,\cdots,4$ clusters}
\label{tab:MC_results}
\vspace{0.5em}
\begin{threeparttable}
\small
\setlength{\tabcolsep}{4pt}
\begin{tabular}{rr|rrrr|c}
\hline
$G$ & $n$ & $K=1$ & $K=2$ & $K=3$ & $K=4$ & \% Correct Classification \\
\hline
100 & 50 & 0.0 & \textbf{91.3} & 8.7 & 0.0 & 98.1\\
200 & 50 & 0.0 & \textbf{96.3} & 3.7 & 0.0 & 98.2\\
100 & 100 & 0.0 & \textbf{95.7} & 4.3 & 0.0 & 99.7\\
200 & 100 & 0.0 & \textbf{99.2} & 0.8 & 0.0 & 99.7\\
100 & 200 & 0.0 & \textbf{99.7} & 0.3 & 0.0 & 100.0\\
200 & 200 & 0.0 & \textbf{99.1} & 0.9 & 0.0 & 100.0\\
\hline
\end{tabular}
\end{threeparttable}
\end{table}
Table~\ref{tab:MC_results} reports the percentage of selecting $K=1,\dots,4$ clusters and the average correct classification rate conditional on selecting $K=2$. The simulation results of classification reveal two key patterns as the number of students per school ($n$) increases: (i) the probability of correctly identifying the number of latent clusters converges to 1; and (ii) classification accuracy improves monotonically.
The simulation scenario with $(100,100)$ mirrors the structure of our empirical dataset, which comprises 119 schools with an average of 78 students each (totaling 9,262 students), and our algorithm determines 2 clusters. In this configuration, the C-Lasso method demonstrates robust performance: it selects the true number of clusters $K_0 = 2$ in 95.7\% of replications, while achieving approximately 99.7\% classification accuracy.
These results provide favorable evidence that our approach is suitable for the empirical application to be presented in Section \ref{sec:empirical_application}.

\begin{table}[!htbp]
\centering
\caption{Simulation Results with Oracle Grouping}
\label{tab:oracle_simulation_results}
\vspace{0.5em}
\begin{threeparttable}
\small
\setlength{\tabcolsep}{4pt}
\begin{tabular}{rr|rr|rrr|rr|rrr}
\hline
& & \multicolumn{5}{c|}{$\overline{\beta}_k$} & \multicolumn{5}{c}{$\beta_k$} \\
\cline{3-12}
& & \multicolumn{2}{c|}{Original} & \multicolumn{3}{c|}{Debiased} & \multicolumn{2}{c|}{Original} & \multicolumn{3}{c}{Debiased} \\
\cline{3-4}\cline{5-7}\cline{8-9}\cline{10-12}
$G$ & $n$ & Bias & RMSE & Bias & RMSE & 95\%CR & Bias & RMSE & Bias & RMSE & 95\%CR \\
\hline
\multicolumn{12}{c}{\textbf{Cluster 1}}  \\
\hline
100 & 50 & 0.009 & 0.172 & -0.015 & 0.170 & 0.956 & -0.024 & 0.059 & 0.001 & 0.053 & 0.962 \\
100 & 100 & 0.012 & 0.123 & 0.005 & 0.122 & 0.941 & -0.012 & 0.040 & 0.000 & 0.038 & 0.950 \\
100 & 200 & 0.007 & 0.084 & 0.003 & 0.084 & 0.945 & -0.005 & 0.027 & 0.000 & 0.026 & 0.945 \\
200 & 50 & 0.017 & 0.124 & -0.005 & 0.121 & 0.958 & -0.024 & 0.046 & -0.000 & 0.038 & 0.962 \\
200 & 100 & 0.003 & 0.083 & -0.006 & 0.082 & 0.949 & -0.009 & 0.029 & 0.003 & 0.027 & 0.951 \\
200 & 200 & 0.002 & 0.059 & -0.002 & 0.059 & 0.959 & -0.007 & 0.020 & -0.001 & 0.019 & 0.944 \\
\hline
\multicolumn{12}{c}{\textbf{Cluster 2}}  \\
\hline
100 & 50 & 0.004 & 0.228 & 0.006 & 0.226 & 0.951 & 0.020 & 0.070 & -0.004 & 0.065 & 0.960 \\
100 & 100 & -0.008 & 0.161 & -0.009 & 0.160 & 0.944 & 0.014 & 0.049 & 0.001 & 0.046 & 0.945 \\
100 & 200 & -0.004 & 0.111 & -0.002 & 0.111 & 0.948 & 0.007 & 0.033 & 0.001 & 0.032 & 0.943 \\
200 & 50 & -0.008 & 0.162 & -0.007 & 0.160 & 0.954 & 0.022 & 0.051 & -0.001 & 0.044 & 0.958 \\
200 & 100 & 0.004 & 0.109 & 0.003 & 0.109 & 0.949 & 0.010 & 0.033 & -0.002 & 0.031 & 0.954 \\
200 & 200 & 0.001 & 0.075 & 0.002 & 0.074 & 0.952 & 0.005 & 0.024 & -0.001 & 0.024 & 0.930 \\
\hline
\end{tabular}
\end{threeparttable}
\end{table}
Since there is no prior literature analyzing incidental parameter bias and its correction in the NPL framework, we assess the finite sample performance using oracle grouping. Oracle grouping ensures that any finite sample bias does not arise from misclassification. We consider median bias (Bias), RMSE, and 95\% coverage rate (95\%CR; reported only for the debiased estimator). Table \ref{tab:oracle_simulation_results} presents the simulation results. The debiased estimator significantly reduces bias, lowers RMSE, and achieves good coverage---especially for $\beta_k$. These results provide evidence of the efficacy of the debiased procedure.

We next evaluate the performance of our post-classification estimators, restricting attention to cases where the adjusted algorithm correctly selects $K = K_0$. We also present simulation results for estimating the unknown slope coefficients using a pooled estimator based on all observations.
Simulation results in Table \ref{tab:simulation_results_k3_simple} show that
the pooled estimator, which fails to capture the heterogeneity across groups, performs poorly. On the other hand, the post-classification estimator $\widehat\theta_{\widetilde{C}_k}$  substantially enhances the finite sample estimation quality,
and the bias is further reduced when using $\widehat\theta_{\widetilde{C}_k}^*$.
Compared with the results in Table~\ref{tab:oracle_simulation_results},
when the true cluster identity is unknown, misclassification error is inevitable in finite samples, which affects the bias and statistical inference when $n=50$.
As soon as $n$ reaches 100, the bias decreases rapidly and the bootstrap coverage rate approaches the nominal 95\%.

Table~\ref{tab:ape_total_results} reports simulation results for the average partial effects of the covariates.
The structure parallels that of the coefficient tables: for each cluster, we report the pooled, original, and debiased estimates.
The pooled estimator, which ignores cluster heterogeneity, produces an economically large bias: it overestimates the APE for Cluster~1 by 0.206 and underestimates it for Cluster~2 by 0.191, reflecting the inability of a common parameter vector to capture the distinct peer effect structures across clusters.
The post-classification estimator reduces the bias to near zero in both clusters as well as for the weighted average, and the debiased procedure further improves coverage.
As $n$ increases, the RMSE of the post-classification estimator declines markedly (e.g., from 0.012 to 0.005 for Cluster~1 at $G=100$) and the coverage rates of the debiased estimates improve toward the theoretical level.
\begin{landscape}
\begin{table}
\centering
\caption{Post-Classification Simulation Results}
\label{tab:simulation_results_k3_simple}
\vspace{0.5em}
\begin{threeparttable}
\small
\setlength{\tabcolsep}{4pt}
\begin{tabular}{rr|rr|rr|rrr|rr|rr|rrr}
\hline
& & \multicolumn{7}{c|}{$\overline{\beta}_k$} & \multicolumn{7}{c}{$\beta_k$} \\
\cline{3-9}\cline{10-16}
& & \multicolumn{2}{c|}{Pooled} & \multicolumn{2}{c|}{Original} & \multicolumn{3}{c|}{Debiased} & \multicolumn{2}{c|}{Pooled} & \multicolumn{2}{c|}{Original} & \multicolumn{3}{c}{Debiased} \\
\cline{3-4}\cline{5-6}\cline{7-9}\cline{10-11}\cline{12-13}\cline{14-16}
$G$ & $n$ & Bias & RMSE & Bias & RMSE & Bias & RMSE & 95\%CR & Bias & RMSE & Bias & RMSE & Bias & RMSE & 95\%CR \\
\hline
\multicolumn{16}{c}{\textbf{Cluster 1}}  \\
\hline
100 & 50 & -1.288 & 1.302 & 0.034 & 0.183 & 0.012 & 0.178 & 0.939 & 0.947 & 0.947 & -0.017 & 0.064 & 0.007 & 0.060 & 0.924 \\
100 & 100 & -1.274 & 1.292 & 0.013 & 0.123 & 0.005 & 0.122 & 0.941 & 0.947 & 0.948 & -0.012 & 0.041 & 0.000 & 0.039 & 0.943 \\
100 & 200 & -1.282 & 1.295 & 0.006 & 0.085 & 0.002 & 0.084 & 0.949 & 0.948 & 0.948 & -0.005 & 0.027 & 0.001 & 0.027 & 0.949 \\
200 & 50 & -1.277 & 1.287 & 0.051 & 0.138 & 0.034 & 0.131 & 0.936 & 0.947 & 0.946 & -0.016 & 0.047 & 0.009 & 0.044 & 0.906 \\
200 & 100 & -1.279 & 1.286 & 0.007 & 0.084 & -0.003 & 0.083 & 0.947 & 0.949 & 0.950 & -0.007 & 0.029 & 0.004 & 0.028 & 0.942 \\
200 & 200 & -1.272 & 1.285 & 0.002 & 0.059 & -0.002 & 0.059 & 0.954 & 0.948 & 0.948 & -0.007 & 0.020 & -0.001 & 0.019 & 0.941 \\
\hline
\multicolumn{16}{c}{\textbf{Cluster 2}}  \\
\hline
100 & 50 & 0.212 & 0.264 & -0.030 & 0.241 & -0.028 & 0.239 & 0.934 & -1.053 & 1.054 & -0.001 & 0.085 & -0.027 & 0.087 & 0.837 \\
100 & 100 & 0.226 & 0.272 & -0.011 & 0.161 & -0.013 & 0.160 & 0.940 & -1.053 & 1.053 & 0.013 & 0.051 & 0.000 & 0.049 & 0.933 \\
100 & 200 & 0.218 & 0.261 & -0.004 & 0.111 & -0.004 & 0.110 & 0.948 & -1.052 & 1.052 & 0.007 & 0.033 & 0.000 & 0.032 & 0.942 \\
200 & 50 & 0.223 & 0.247 & -0.049 & 0.175 & -0.049 & 0.173 & 0.936 & -1.053 & 1.054 & -0.011 & 0.066 & -0.035 & 0.072 & 0.786 \\
200 & 100 & 0.221 & 0.243 & -0.001 & 0.110 & -0.000 & 0.110 & 0.952 & -1.051 & 1.051 & 0.007 & 0.036 & -0.005 & 0.035 & 0.927 \\
200 & 200 & 0.228 & 0.246 & 0.001 & 0.075 & 0.001 & 0.075 & 0.954 & -1.052 & 1.052 & 0.005 & 0.025 & -0.000 & 0.024 & 0.932 \\
\hline
\end{tabular}
\end{threeparttable}
\end{table}
\end{landscape}

\begin{table}[!ht]
\centering
\caption{Post-Classification Simulation Results for Average Partial Effects}
\label{tab:ape_total_results}
\vspace{0.5em}
\begin{threeparttable}
\small
\setlength{\tabcolsep}{4pt}
\begin{tabular}{rr|rr|rr|rrr}
\hline
& & \multicolumn{2}{c|}{Pooled} & \multicolumn{2}{c|}{Original} & \multicolumn{3}{c}{Debiased} \\
\cline{3-4}\cline{5-6}\cline{7-9}
$G$ & $n$ & Bias & RMSE & Bias & RMSE & Bias & RMSE & 95\%CR \\
\hline
\multicolumn{9}{c}{\textbf{Cluster 1}} \\
\hline
  100 & 50 & 0.206 & 0.206 & 0.001 & 0.012 & 0.001 & 0.012 & 0.896 \\
  100 & 100 & 0.206 & 0.207 & 0.000 & 0.008 & 0.000 & 0.008 & 0.928 \\
  100 & 200 & 0.207 & 0.207 & 0.000 & 0.005 & 0.000 & 0.005 & 0.931 \\
  200 & 50 & 0.206 & 0.206 & 0.001 & 0.008 & 0.001 & 0.008 & 0.902 \\
  200 & 100 & 0.207 & 0.207 & 0.001 & 0.005 & 0.001 & 0.005 & 0.930 \\
  200 & 200 & 0.207 & 0.206 & -0.000 & 0.004 & 0.000 & 0.004 & 0.937 \\
\hline
\multicolumn{9}{c}{\textbf{Cluster 2}} \\
\hline
  100 & 50 & -0.191 & 0.192 & -0.003 & 0.014 & -0.003 & 0.014 & 0.854 \\
  100 & 100 & -0.191 & 0.191 & -0.000 & 0.008 & -0.000 & 0.008 & 0.928 \\
  100 & 200 & -0.191 & 0.191 & 0.000 & 0.005 & 0.000 & 0.005 & 0.950 \\
  200 & 50 & -0.191 & 0.192 & -0.005 & 0.011 & -0.005 & 0.011 & 0.791 \\
  200 & 100 & -0.191 & 0.191 & -0.000 & 0.006 & -0.000 & 0.006 & 0.918 \\
  200 & 200 & -0.191 & 0.191 & 0.000 & 0.004 & 0.000 & 0.004 & 0.930 \\
\hline
\multicolumn{9}{c}{\textbf{Weighted Average}} \\
\hline
  100 & 50 & 0.047 & 0.048 & -0.001 & 0.008 & -0.001 & 0.008 & 0.929 \\
  100 & 100 & 0.047 & 0.048 & -0.000 & 0.005 & -0.000 & 0.005 & 0.939 \\
  100 & 200 & 0.048 & 0.048 & 0.000 & 0.004 & 0.000 & 0.004 & 0.942 \\
  200 & 50 & 0.047 & 0.047 & -0.001 & 0.005 & -0.001 & 0.005 & 0.947 \\
  200 & 100 & 0.048 & 0.048 & 0.000 & 0.004 & 0.000 & 0.004 & 0.955 \\
  200 & 200 & 0.048 & 0.048 & -0.000 & 0.003 & -0.000 & 0.003 & 0.955 \\
\hline
\end{tabular}
\end{threeparttable}
\end{table}

The weighted average panel of Table~\ref{tab:ape_total_results} presents the simulation results for the weighted APE. The pooled estimator has a non-negligible bias and the bias does not decrease with either $G$ or $n$, which demonstrates that the pooled inference is biased even though the target is not cluster-specific. In contrast, the debiased estimator is nearly unbiased with much smaller RMSE, which is decreasing with $G$ and $n$. The coverage rate of the bootstrap confidence interval approaches the nominal level as $n$ increases. The weighted average cancels out the cluster-level biases arising from the misclassification. So the CR of the bootstrap confidence interval of the weighted APE is closer to the nominal level than that of the cluster-specific APEs.

Overall, the APE results reinforce the findings from the coefficient tables: the proposed method delivers reliable inference on economically meaningful quantities once cluster heterogeneity is properly accounted for.

The findings of this simulation exercise corroborate the asymptotic results presented in the previous section. The algorithm is capable of identifying the cluster identities of the groups, and the post-classification estimators effectively approximate the oracle performance as if the cluster membership were known.

In terms of the computational costs of this comprehensive simulation exercise, in total we ran 12 million bootstrap replicates.\footnote{There are 12 DGPs, including 2 scenarios with 6 $(n,G)$ combinations and 2 clusters. For each DGP, we replicate 1,000 times with 500 bootstrap replications.} We parallelize the process and allocate the tasks to 128 CPU cores (AMD EPYC 9754) in a high-performance computing cluster. It took around 20 hours to finish the simulations. 
The following empirical application, with no need of replication, is fast. Using the same parallel scheme, it took fewer than 20 seconds. Therefore, the computation time is affordable for real-data applications.

\section{Heterogeneous Peer Effects in Risky Behavior} \label{sec:empirical_application}
 \noindent Many empirical studies demonstrate that risky behaviors, such as smoking, binge drinking, or unsafe sexual practices, can spread across social networks \citep{christakis2007spread,centola2010spread}. Peer effects can play a key role in this process, especially among students; see, for example, \citet{duncan2005peer} and \citet{eisenberg2014peer}. However, the existing literature usually focuses on homogeneous peer effects. Understanding heterogeneous peer effects can help policymakers design targeted public interventions for distinct clusters. In this section, we apply our proposed procedure to
 the well-known Add Health dataset
to study heterogeneous peer effects in risky behaviors among students. 

\subsection{Add Health Dataset}

The Add Health study is a longitudinal study of a nationally representative sample of over 20,000 adolescents who were in grades 7--12 during the 1994--95 school year, and have been followed for five waves to date, most recently in 2016--18.\footnote{See detailed description in \url{https://addhealth.cpc.unc.edu/}.} It is widely used in the study of peer effects.
Table \ref{summary_statistics_risky} provides summary statistics of the students and schools.  The dataset contains 119 schools and 9,262 students. The dependent variable, \emph{risky behavior}, is constructed by using the self-report questionnaires in the Add Health dataset. Specifically, the survey question is 
 
 ``\textit{During the past twelve months, how often did you do something dangerous because you were dared to?}''

\begin{table}[htbp]
\centering
\caption{Summary Statistics}
\vspace{0.5em}
\label{summary_statistics_risky}
\begin{tabular}{rrS[table-format=2.3]S[table-format=2.3]S[table-format=2.0]S[table-format=3.3]}
\hline
                              & {Obs.} & {Mean}   & {SD}     & {Min} & {Max}   \\ \hline
\underline{\textbf{Outcome}}                      &      &        &        &     &       \\
Risky Behavior                & 9,262 & 0.409  & 0.492  & 0   & 1     \\
\underline{\textbf{Student Characteristics}}     &      &        &        &     &       \\
Number of Friends & 9,262  & 1.048 & 1.506 & 0  & 10  \\
Intelligence (3--4)       & 9,262 & 0.593  & 0.491  & 0   & 1     \\
Intelligence (5--6)       & 9,262 & 0.359  & 0.480  & 0   & 1     \\
White                         & 9,262 & 0.612  & 0.487  & 0   & 1     \\
Female                        & 9,262 & 0.519  & 0.500  & 0   & 1     \\
(Log) Income                  & 9,262 & 3.578  & 0.804  & 0   & 6.907 \\
Age                           & 9,262 & 15.606 & 1.643  & 12  & 21    \\
GPA                           & 9,262 & 2.825  & 0.797  & 1   & 4     \\
\hline
\end{tabular}
\end{table}

Demographic characteristics include intelligence level (low (1--2), middle (3--4), high (5--6); with the low category as the reference group), race, gender, family income, age, grades, and in-school friendship networks. There are 119 schools, with an average of 78 students per school. 
On average, each student has only one reported friend, suggesting sparse friendship networks. The share of students with a low intelligence level is much lower than the other two intelligence levels.

\subsection{Empirical Results}

We first estimate the peer effects and coefficients of covariates under the homogeneous model as a benchmark. The first column in Table \ref{tab:estimation_results_combined} presents the results.  The point estimate of peer effects is $0.303$ with a 95\% confidence interval $[0.083,0.515]$, where the confidence interval is estimated by the parametric bootstrap. This demonstrates that, with other characteristics controlled, peer influence increases the probability of students engaging in risky behaviors. The influence of students' intelligence levels is insignificant. White students and those with higher family incomes are relatively more inclined to participate in risky activities. Female students, older students, and those with higher GPAs are less likely to engage in risky behaviors.

The first column in Table~\ref{tab:ape_empirical} presents the pooled estimation of the APE of each covariate on the probability of engaging in risky behavior. The results demonstrate that being White and being Female have the largest effects on the CCPs.

\begin{table}[ht]
\centering
\caption{Estimation Results with Homogeneous/Heterogeneous Clusters }
\vspace{0.5em}
\label{tab:estimation_results_combined}
\begin{threeparttable}
\begin{tabular}{rc|cc}
\hline
          & Homogeneous & \multicolumn{2}{c}{Heterogeneous}              \\
Variables & Pooled    & First Cluster & Second Cluster \\
\hline
\addlinespace
Peer Effects & 0.303 & -0.225 & 0.688 \\
 & [0.083, 0.515] & [-0.558, 0.125] & [0.406, 0.943] \\
 \addlinespace
Intelligence (3--4) & 0.098 & -0.344 & 0.419 \\
 & [-0.094, 0.263] & [-0.658, -0.059] & [0.133, 0.669] \\
 \addlinespace
Intelligence (5--6) & -0.063 & -0.621 & 0.347 \\
 & [-0.264, 0.127] & [-0.958, -0.305] & [0.046, 0.632] \\
 \addlinespace
White & 0.323 & 0.085 & 0.495 \\
 & [0.202, 0.438] & [-0.092, 0.252] & [0.345, 0.646] \\
 \addlinespace
Female & -0.761 & -0.754 & -0.752 \\
 & [-0.832, -0.656] & [-0.880, -0.617] & [-0.864, -0.630] \\
 \addlinespace
(Log) Income & 0.092 & 0.172 & 0.008 \\
 & [0.026, 0.147] & [0.075, 0.263] & [-0.073, 0.084] \\
 \addlinespace
Age & -0.015 & -0.040 & -0.024 \\
 & [-0.036, -0.000] & [-0.066, -0.014] & [-0.048, 0.002] \\
 \addlinespace
GPA & -0.157 & -0.102 & -0.206 \\
 & [-0.212, -0.094] & [-0.187, -0.015] & [-0.284, -0.125] \\
 \addlinespace
\hline
\end{tabular}
\begin{tablenotes}
\item \footnotesize{
\textbf{Note}: 95\% confidence intervals in the brackets are estimated by parametric bootstrap.
}
\end{tablenotes}
\end{threeparttable}
\end{table}

\begin{table}[ht]
\centering
\caption{Average Partial Effects of Covariates}
\vspace{0.5em}
\label{tab:ape_empirical}
\begin{threeparttable}
\begin{tabular}{rc|cccc}
\hline
           & Homogeneous & \multicolumn{3}{c}{Heterogeneous} \\
Variables  & Pooled      & Weighted & First Cluster & Second Cluster \\
\hline
\addlinespace
Intelligence (3--4) & 0.009 & 0.019 & -0.025 & 0.044 \\
 & [-0.006, 0.028] & [-0.003, 0.037] & [-0.048, 0.000] & [0.020, 0.075] \\
\addlinespace
Intelligence (5--6) & -0.010 & 0.011 & -0.075 & 0.059 \\
 & [-0.039, 0.022] & [-0.024, 0.043] & [-0.120, -0.037] & [0.017, 0.110] \\
\addlinespace
White & 0.029 & 0.035 & 0.006 & 0.051 \\
 & [0.019, 0.040] & [0.023, 0.045] & [-0.005, 0.020] & [0.037, 0.068] \\
\addlinespace
Female & -0.088 & -0.086 & -0.080 & -0.089 \\
 & [-0.102, -0.081] & [-0.092, -0.071] & [-0.103, -0.073] & [-0.106, -0.079] \\
\addlinespace
(Log) Income & 0.010 & 0.008 & 0.021 & 0.001 \\
 & [0.004, 0.017] & [0.001, 0.014] & [0.012, 0.034] & [-0.007, 0.010] \\
\addlinespace
Age & -0.002 & -0.002 & -0.003 & -0.002 \\
 & [-0.003, 0.001] & [-0.005, -0.001] & [-0.004, 0.002] & [-0.005, 0.001] \\
\addlinespace
GPA & -0.017 & -0.018 & -0.011 & -0.022 \\
 & [-0.025, -0.011] & [-0.025, -0.011] & [-0.021, -0.001] & [-0.031, -0.014] \\
\addlinespace
\hline
\end{tabular}
\begin{tablenotes}
\item \footnotesize{
\textbf{Note}: 95\% confidence intervals in the brackets are estimated by parametric bootstrap.
}
\end{tablenotes}
\end{threeparttable}
\end{table}

We have shown that ignoring heterogeneity results in substantial bias in simulations. We apply the adjusted algorithm to identify the latent structure and estimate the peer effects with heterogeneous clusters. The C-Lasso procedure classifies the 119 schools into two clusters, one containing 43 schools and the other 76 schools. Table \ref{summary_statistics_risky_group} presents the summary statistics for these two clusters, respectively. We calculate the two-sample $t$-statistics to examine whether there are statistically significant differences in the mean values of demographic characteristics. We find that while risky behaviors of students in different clusters are similar, some demographic characteristics vary significantly. The first cluster has a higher proportion of White and male students. The students from the first cluster are generally older, have more friends, come from families with higher income, and yet have a relatively lower average GPA. Notably, there is no significant difference in intelligence levels between the two clusters. Moreover, the average number of students in the first cluster (93) is higher than that in the second cluster (69); \citet*{lewbel2023social} document the influence of the group size on peer effects.
\begin{table}[!ht]
\centering
\caption{Comparison of Summary Statistics between Two Clusters}
\vspace{0.5em}
\label{summary_statistics_risky_group}
\resizebox{1\textwidth}{!}{
\begin{threeparttable}
\begin{tabular}{rrS[table-format=2.3]S[table-format=2.3]rS[table-format=2.3]S[table-format=2.3]S[table-format=2.3]S[table-format=3.3]}
\hline
                              & \multicolumn{3}{c}{Cluster 1 (43 schools) } & \multicolumn{3}{c}{Cluster 2 (76 schools)} &   Comparison           \\
                               & {Obs.}     & {Mean}      & {SD}       & {Obs.}      & {Mean}      & {SD}       & {Two-sample $t$-stat.} \\ \hline
\underline{\textbf{Outcome}}                       &          &           &          &           &           &          &              \\
Risky Behavior                & 4,019    & 0.414     & 0.493    & 5,243     & 0.406     & 0.491    & 0.786        \\
\underline{\textbf{Characteristics}}     &          &           &          &           &           &          &              \\
Number of Friends             & 4,019    & 1.245     & 1.681    & 5,243     & 0.897     & 1.337    & 11.102***    \\
Intelligence (3--4)       & 4,019    & 0.590      & 0.492    & 5,243     & 0.595     & 0.491    & -0.485       \\
Intelligence (5--6)       & 4,019    & 0.363     & 0.481    & 5,243     & 0.356     & 0.479    & 0.697        \\
White                         & 4,019    & 0.657     & 0.475    & 5,243     & 0.577     & 0.494    & 7.865***     \\
Female                        & 4,019    & 0.506     & 0.500    & 5,243     & 0.529     & 0.499    & -2.219**     \\
(Log) Income                  & 4,019    & 3.607     & 0.831    & 5,243     & 3.555     & 0.783    & 3.076***     \\
Age                           & 4,019    & 15.748    & 1.579    & 5,243     & 15.497    & 1.682    & 7.310***     \\
GPA                           & 4,019    & 2.804     & 0.805    & 5,243     & 2.841     & 0.791    & -2.219**     \\\hline
\end{tabular}
\begin{tablenotes}
\item \footnotesize{Significance Level: *** 1\%, ** 5\%, * 10\%.}
\end{tablenotes}
\end{threeparttable}
}
\end{table}

The heterogeneous panel of Table \ref{tab:estimation_results_combined} provides the post-classification estimation results of two clusters separately.
The point estimates for the two clusters exhibit distinct patterns. The most significant finding is that only the point estimate of peer effects in the second cluster, 0.688, is both statistically significant and economically meaningful. 
The peer effects for the first cluster are not statistically different from zero.
In addition to the differences in peer effects, students with higher intelligence levels in the second cluster are more inclined to make risky choices, whereas students in the first cluster are not. Female students or those with higher grades are generally less likely to engage in risky behaviors in both clusters, which aligns with economic intuition.

If we ignore the latent structures, we may conclude that intelligence does not affect students' risky behaviors. However, intelligence significantly affects both the first and the second clusters with opposite signs. Moreover, we can obtain a more precise estimate of the peer effects for groups in the second cluster, which is helpful for policymaking.

The APE results in Table \ref{tab:ape_empirical} reinforce the coefficient-based findings: the intelligence variables exhibit opposite signs in the two clusters. For policy considerations, the APEs are more interpretable than the coefficients. For example, in the second cluster, being White raises the probability of risky behavior by 5.1 percentage points on average, while being Female reduces it by 8.9 percentage points.

\section{Conclusion}\label{sec:conclusion}
In this paper, we study social interactions models with latent structures. We model the heterogeneity across groups via group fixed effects and cluster-specific coefficients. By combining NPL and C-Lasso, our proposed estimation method can identify the latent structures and unknown parameters simultaneously. We illustrate the finite sample performance of the method by Monte Carlo experiments and an application to heterogeneous peer effects in risky behavior using the Add Health dataset. Schools are classified into two clusters, and only one of them has statistically significant peer effects. Researchers should be aware of the latent structures when studying social interactions models. 

\bigskip
\bigskip


\newpage
\appendix
\section*{\centering\LARGE{Appendices}}
The appendices comprise two parts. In Appendix \ref{appendix}, we introduce some technical assumptions and lemmas and present the proofs of the main results. We then prove the lemmas in Appendix \ref{appendix-lemmas}.

The proofs call for some new notation, which we define here.  Denote by $C$ a generic fixed positive constant.
Throughout the proofs, we use the notation "$\ \breve{\cdot}\ $"  above a symbol to represent a generic mean value between the estimator and the true value. We write $a_n=O(b_n)$ if $a_n<Cb_n$  and $a_n=o(b_n)$ if $a_n=c_nb_n$ for some $c_n\to 0$ as $n\to\infty$. $O_p(\cdot)$ and $o_p(\cdot)$  are defined for the corresponding limit relationships in probability.  "$\plim$" represents the limit in probability.
 For any real matrix $A$, we let $\|A\|_1=\max_{j}|\sum_{i}A_{ij}|$, $\|A\|_2=\sqrt{\lambda_{\max}(A' A)}$, $\|A\|_\infty=\max_{i}|\sum_jA_{ij}|$, and $\|A\|_F=\sqrt{\mathrm{Tr}(A' A)}$, where $\lambda_{\max}(\cdot)$  and $\mathrm{Tr}(\cdot)$ are the largest eigenvalue and the trace, respectively. For a function $f(x_1,\dots,x_d)$, let $\partial^s f(x_1,\dots,x_d)/\partial x_{1}^{s_1} \dots \partial x_{d}^{s_d}$, where $\sum_{i=1}^d s_i = s$, denote the $s$th partial derivative of $f$. We denote this as $\partial^s f_0/\partial x_1^{s_1} \dots \partial x_d^{s_d}$ when $(x_1,\dots,x_d)$ is evaluated at the true value $(x_{10},\dots,x_{d0})$.

\section{Proofs of Main Results}\label{appendix}
\renewcommand{\theequation}{A.\arabic{equation}}
\setcounter{equation}{0}
\renewcommand{\thetheorem}{A.\arabic{theorem}}
\setcounter{theorem}{0}
\renewcommand{\thelemma}{A.\arabic{lemma}}
\setcounter{lemma}{0}
\renewcommand{\theassumption}{A.\arabic{assumption}}
\setcounter{assumption}{0}
In this Appendix, we prove the main theoretical results in Section \ref{sec:large_sample_results}. Proposition \ref{prop:first_step_npl}, Theorem \ref{thm:post_classification_estimation}, Theorem \ref{thm:classification_consistency}, Proposition \ref{prop:K-selection}, and Corollary \ref{cor:ape_asymptotics} are developed in Appendices \ref{appendix:first-NPL}--\ref{appendix:ape_corollary}, respectively.

\subsection{Proof of Proposition \ref{prop:first_step_npl}}\label{appendix:first-NPL}
Because groups are independent and have the same size, we omit the (sub)superscript $g$ to simplify the notation if possible. 

Let $P(\zeta)$ be the fixed point of the equation $P = \Gamma_n(\zeta, P)$. We first prove that $P(\zeta)$ exists and is unique when $\|\zeta\|_2 < \infty$ and $|\overline\beta| \in (0,4)$. 
\begin{lemma}\label{lem:contraction_mapping}
Suppose Assumption \ref{npl} holds, $\|\zeta\| < \infty$, and $|\overline\beta| \in (0,4)$. Then the equation $P = \Gamma_n(\zeta, P)$ has a unique solution $P(\zeta)$. Additionally, $P(\zeta)$ satisfies $$\left\|\frac{\partial P(\zeta)}{\partial \zeta}\right\|_2\leq C\sqrt n.$$
\end{lemma}

\medskip

The next lemma presents two deviation bounds.
\begin{lemma}\label{lem:dev_bounds}
    Under Assumption \ref{npl}, we have 
    \begin{equation}
    \label{eq:lem2-bound1}
        \left\|\frac{1}{n}\sum_{i\in[n]}\frac{\partial\psi(w_i, P_0;\zeta_0)}{\partial\zeta}\right\|_\infty \leq C\sqrt{\frac{\log(Gn)}{n}},
    \end{equation}
    and
    \begin{equation}
    \label{eq:lem2-bound2}
        \left|\frac{1}{n}\sum_{i\in[n]}\{\psi(w_i,P;\zeta)-\mathbb{E}[\psi(w_i,P;\zeta)]\}\right| \leq C\sqrt{\frac{\log(Gn)}{n}},
    \end{equation}
   uniformly over $P$ and $\zeta\in\mathcal{A}\times\Theta$ with probability at least $1-2G^{-2}n^{-2}$.
\end{lemma}

\medskip

To proceed, we need the following invertibility assumption.
\begin{assumption}
    \label{assu:invertibility_assumption1}
    For sufficiently large $n$, the following matrices are invertible w.p.a.1.,
    $$
    \frac{1}{n}\sum_{i\in[n]}\left[\frac{\partial^2\psi(w_{ig} ,P_g(\zeta);\zeta)}{\partial\zeta\partial\zeta'}+\frac{\partial^2\psi(w_{ig}, P_g(\zeta);\zeta)}{\partial\zeta\partial P'}\frac{\partial P_g(\zeta)}{\partial\zeta'}\right],\quad g\in[G].
    $$
\end{assumption}

\medskip
Let  $\Gamma_{ng}(\zeta_g;P_g)$ be an $n$-dimensional vector with each element
\[
[\Gamma_{ng}(\zeta_g;P_g)]_i=\Lambda\Big(\overline\beta_g\overline{P}_{ig}+X_{ig}'\beta_g+\mu_g\Big),\quad i\in[n],g\in[G].
\]
Its population version is 
\[
[\Gamma_{g}(\zeta_g;P_g)]_i=\Lambda\Big(\overline\beta_g\overline{P}_{ig}+X_{ig}'\beta_g+\mu_g\Big),\quad i\in[n],g\in[G].
\]
We define the population and sample NPL fixed point operators, respectively:
\begin{equation}
\label{eq:NPL_fixed_point_population}
     \phi_{g}(P_g)=\Gamma_g(\widetilde \zeta_g(P_g);P_g),\quad \text{with}~~\widetilde\zeta_{g}(P_g)= \underset{\zeta_g\in\mathcal{A}\times\Theta}{\arg\min}~\Psi_{g}(\zeta_g;P_g),
\end{equation}
and 
\begin{equation}
\label{eq:NPL_fixed_point_sample}
    \phi_{ng}(P_g)=\Gamma_{ng}(\widetilde\zeta_g(P_g);P_g),\quad \text{with}~~\widetilde\zeta_g(P_g)= \underset{\zeta_g\in\mathcal{A}\times\Theta}{\arg\min}~\Psi_{ng}(\zeta_g;P_g),
\end{equation} 
where $\Psi_{g}(\zeta_g;P_g)=\mathbb{E}[\Psi_{ng}(\zeta_g;P_g)]$.
We define the population NPL fixed points set as 
\[
\Lambda_{g0}=\{(\zeta_g,P_g)\in(\mathcal{A}\times\Theta\times\mathcal{P}_n):\zeta_g=\widetilde\zeta_g(P_g),P_g=\phi_g(P_g)\},
\]
and the NPL fixed points set of sample size $n$ as 
\[
\Lambda_{ng}=\{(\zeta_g,P_g)\in(\mathcal{A},\Theta,\mathcal{P}_n):\zeta_g=\widetilde\zeta_g(P_g),P_g=\phi_{ng}(P_g)\},
\] 
where $\mathcal{P}_n=(\underline{c},1-\underline{c})^n$. These definitions will be used in the proof of the consistency of $\widetilde{\zeta}_g$.

Now we are ready for the proof of Proposition \ref{prop:first_step_npl}.

\begin{proof}[Proof of Proposition \ref{prop:first_step_npl}]
Part \ref{first-step-consistency}.  Define a desirable event for $g\in[G]$, $$D_g=\Big\{\Big|\Psi_{ng}(\zeta_g;P_g)-\Psi_{g}(\zeta_g;P_g)\Big| \leq C\sqrt{\frac{\log(Gn)}{n}}, \ \text{uniformly for $(P_g,\zeta_g)\in\mathcal{A}\times\Theta\times\mathcal{P}_n$}\Big\}$$  and the intersection of these events $D=\cap_{g\in[G]}D_g.$
By Lemma \ref{lem:dev_bounds}, we have $\Pr(D_g)\geq 1-2G^{-2}n^{-2}$. The union bound implies
$$
\Pr(D)=1-\Pr(\cup_{g\in[G]}D_g^c)\geq 1-\sum_{g\in[G]}\Pr(D_g^c)\geq 1-G\times(2G^{-2}n^{-2})=1-2G^{-1}n^{-2}.
$$ 

The following proof is conditioned on the desirable event $D$. By Assumption \ref{npl}\ref{npl-ii} and \ref{npl-iv}, we know that $\zeta_{g0}$ uniquely maximizes the population likelihood. Define the function $$\mathcal{T}(\zeta_g,P_g)=\Psi_g(\zeta_g;P_g)-\min_{z\in\mathcal{A}\times\Theta}\{\Psi_g(z;P_g)\}.$$
Because $\Psi_g(\zeta_g;P_g)$ is continuous and the parameter space is compact, Berge's maximum theorem establishes that $\mathcal{T}(\zeta_g,P_g)$ is a continuous function. By construction, $\mathcal{T}(\zeta_g,P_g)\geq 0$ for any $(\zeta_g,P_g)$. Let $\mathcal{F}$ be the set of vectors $(\zeta_g,P_g)$ that are fixed points of the equilibrium mapping $\Gamma_g$, i.e., $$\mathcal{F}=\{(\zeta_g,P_g)\in\mathcal{A}\times\Theta\times\mathcal{P}_n:P_g=\Gamma_g(\mu_g,\theta_g;P_g)\}.$$ Given that $\mathcal{A}\times\Theta\times\mathcal{P}_n$ is compact and $\Gamma_g$ is continuous, $\mathcal{F}$ is then a compact set. Let $$B_\epsilon(\zeta_{g0},P_g^*)=\{(\zeta_g,P_g)\in\mathcal{A}\times\Theta\times\mathcal{P}_n:\|\zeta_g-\zeta_{g0}\|+\|P_g-P_g^*\|<\epsilon\}$$ be an arbitrarily small open ball that contains $(\zeta_{g0},P_g^*)$. Thus $B_\epsilon^c(\zeta_{g0},P_g^*)\cap\mathcal{F}$ is also compact. Define the constant $$\tau=\min_{(\zeta_g,P_g)\in B_\epsilon^c(\zeta_{g0},P_g^*)\cap\mathcal{F}}\mathcal{T}(\zeta_g,P_g).$$ 
When $(\zeta_g,P)\in B_\epsilon^c(\zeta_{g0},P_g^*)$, Assumption \ref{npl}\ref{npl-ii} implies that $\mathcal{T}(\zeta_g,P_g)>0$ for all $\zeta_g\neq\zeta_{g0}$. Therefore $\tau>0$. 

Define the event 
$$
\mathcal{E}_g=\left\{|\Psi_{ng}(\zeta_g;P_g)-\Psi_g(\zeta_g;P_g)|<\frac{\tau}{2}\text{ for all $(\zeta_g,P_g)\in\mathcal{A}\times\Theta\times\mathcal{P}_n$}\right\}. 
$$
Let $(\zeta_g^{(n)},P_g^{(n)})$ be an element of $\Lambda_{ng}$. On the event $\mathcal{E}_g$, we have $$\Psi_g(\zeta_g^{(n)};P_g^{(n)})<\Psi_{ng}(\zeta_g^{(n)};P_g^{(n)})+\frac{\tau}{2},$$ and $$\Psi_{ng}(\zeta_g;P_g^{(n)})<\Psi_g(\zeta_g;P_g^{(n)})+\frac{\tau}{2}.$$ 
Furthermore, from the definition of NPL fixed point, $\zeta_g^{(n)}$ minimizes $\Psi_{ng}(\cdot, P_g^{(n)})$, so $$\Psi_{ng}(\zeta_g^{(n)};P_g^{(n)})\leq \Psi_{ng}(\zeta_g;P_g^{(n)}).$$ 
Therefore, combining these three inequalities yields for any $\zeta_g\in\mathcal{A}\times\Theta$,
$$\Psi_g(\zeta_g^{(n)};P_g^{(n)})<\Psi_g(\zeta_g;P_g^{(n)})+\tau.$$ 
Minimizing the right-hand side over $\zeta_g$ gives
$$\Psi_g(\zeta_g^{(n)};P_g^{(n)})<\min_{z\in\mathcal{A}\times\Theta}\Psi_g(z;P_g^{(n)})+\tau,$$
which is equivalent to $\tau>\mathcal{T}(\zeta_g^{(n)},P_g^{(n)})$ by definition of $\mathcal{T}$.
Thus we make the following deductions: 
$$
\begin{aligned}
    \mathcal{E}_g\Rightarrow\ &\{\tau>\mathcal{T}(\zeta_g^{(n)},P_g^{(n)})\}\\
    \Rightarrow\  &\{\min_{(\zeta_g,P_g)\in B_\epsilon^c(\zeta_{g0},P_g^*)\cap\mathcal{F}}\mathcal{T}(\zeta_g,P_g)>\mathcal{T}(\zeta_g^{(n)},P_g^{(n)})\}\\
    \Rightarrow\ &\{(\zeta_g^{(n)},P_g^{(n)})\in B_\epsilon(\zeta_{g0},P_g^*)\},
\end{aligned}
$$
where the last deduction uses the fact that $(\zeta_g^{(n)},P_g^{(n)})\in\mathcal{F}$. Therefore, $\Pr(\mathcal{E}_g)\leq \Pr((\zeta_g^{(n)},P_g^{(n)})\in B_\epsilon(\zeta_{g0},P_{g0})).$ Notice that the event $D_g$ always implies $\mathcal{E}_g$ for large enough $n$ (i.e., let $C\sqrt{\log(Gn)/n}\leq\tau$). Hence, we have $\Pr(\mathcal{E}_g)\geq \Pr(D_g)$ as $n\to\infty$.  Moreover, as $\epsilon$ is an arbitrarily small constant, we have
$$\max_{g\in[G]}\sup_{(\zeta_g^{(n)},P_g^{(n)})\in\Lambda_{ng}}\|\zeta_g^{(n)}-\zeta_{g0}\|=o_p(1).$$ From the definition of $\Lambda_{ng}$, we see that $\max_{g\in[G]}\|\widetilde\zeta_g-\zeta_{g0}\|=o_p(1)$.  

\medskip
\noindent
Part \ref{first-step-probability-vector}. Define another desirable  event $$E_g=\Big\{\Big\|\frac{1}{n}\sum_{i\in[n]}\frac{\partial\psi(w_{ig}, P_{g0};\zeta_{g0})}{\partial\zeta}\Big\|_\infty\leq C\sqrt{\frac{\log(Gn)}{n}}\Big\},\quad g\in[G],$$ and the intersection $E=\cap_{g\in[G]}E_g.$ 
Similarly, we can apply Lemma \ref{lem:dev_bounds} and the union bound to show that $\Pr(E)\geq 1-2G^{-1}n^{-2}$. The following proof is conditioned on the event $E$.

The first-order condition (FOC) of the NPL estimator is
\begin{equation*}
    \frac{1}{n}\sum_{i\in[n]}\frac{\partial\psi(w_i, P(\widetilde\zeta);\widetilde{\zeta})}{\partial \zeta}=0,
\end{equation*}
where $P(\widetilde{\zeta})$ is the unique solution to $P=\Gamma(\widetilde{\zeta},P)$. By a Taylor expansion around the true value $\zeta_0$, we have 
\begin{equation*}
    \begin{aligned}
        \frac{1}{n}\sum_{i\in[n]}&\frac{\partial\psi(w_i,P(\widetilde{\zeta});\widetilde{\zeta})}{\partial \zeta}
        =\frac{1}{n}\sum_{i\in[n]}\frac{\partial\psi(w_{i} ,P(\zeta_0);\zeta_0)}{\partial\zeta}\\
        &+\frac{1}{n}\sum_{i\in[n]}\left[\frac{\partial^2\psi(w_{i} ,P(\breve\zeta);\breve\zeta)}{\partial\zeta\partial\zeta'}+\frac{\partial^2\psi(w_i, P(\breve\zeta);\breve\zeta)}{\partial\zeta\partial P'}\frac{\partial P(\breve\zeta)}{\partial\zeta'}\right](\widetilde\zeta-\zeta_0),
    \end{aligned}
\end{equation*}
where $\breve{\zeta}$ is a general mean value.\footnote{Since the mean value may differ across rows when writing the Taylor expansion for a vector-valued function, we slightly abuse the notation for conciseness.}
Hence 
\begin{align*}
    \widetilde{\zeta}-\zeta_0 =-&\left\{\frac{1}{n}\sum_{i\in[n]}\left[\frac{\partial^2\psi(w_{i} ,P(\breve\zeta);\breve\zeta)}{\partial\zeta\partial\zeta'}+\frac{\partial^2\psi(w_i, P(\breve\zeta);\breve\zeta)}{\partial\zeta\partial P'}\frac{\partial P(\breve\zeta)}{\partial\zeta'}\right]\right\}^{-1}\\&\times\frac{1}{n}\sum_{i\in[n]}\frac{\partial\psi(w_i,P_0;\zeta_0)}{\partial\zeta},
\end{align*}
by Assumption \ref{assu:invertibility_assumption1}
and the fact $P(\zeta_0)=P_0$. Therefore, under the desirable event $E$, $\|\widetilde{\zeta}-\zeta_0 \|_2\leq C\sqrt{\log(Gn)/n}$. Moreover, by Lemma \ref{lem:contraction_mapping},  we have $\left\|\frac{\partial P(\zeta)}{\partial\zeta}\right\|_2\leq C\sqrt{n}$. Hence,  $$\left\|\widetilde{P}-P_0\right\|_2\leq \left\|\frac{\partial P(\zeta)}{\partial\zeta}\right\|_2\|\widetilde\zeta-\zeta_0\|_2\leq C\sqrt{n}\sqrt{\frac{\log(Gn)}{n}}=C\sqrt{\log(Gn)},$$ and $\max_{g\in[G]}\left\|\widetilde{P}_g-P_{g0}\right\|_2 \leq C\sqrt{\log(Gn)}$ follows.
\end{proof}

\subsection{Proof of Theorem \ref{thm:post_classification_estimation}}\label{appendix:post-estimation}
Before the proof, we state some additional technical conditions to ensure that the matrices or quantities involved in the proof are invertible or nonzero.
\begin{assumption}
\label{assu:invertibility2}
    For sufficiently large $n$ and $G$, 
    \begin{enumerate}[label=(\roman*),font=\upshape]
        \item\label{inver-cond1} The expression $$\frac{1}{n}\sum_{i\in[n]}\Big[\frac{\partial^2\psi(w_{ig}, P_g(\widehat\mu_g(\theta),\theta);\widehat\mu_g(\theta),\theta)}{\partial\mu_g^2}+\frac{\partial^2\psi(w_{ig}, P_g(\widehat\mu_g(\theta),\theta);\widehat\mu_g(\theta),\theta)}{\partial\mu_g\partial P_g'}\frac{\partial P_g(\widehat\mu_g(\theta),\theta)}{\partial\mu_g}\Big]$$ is nonzero for $\theta\in\Theta$ and $g\in[G]$,
        \item\label{inver-cond2} $\Sigma_{\mu_g\mu_g0}$ defined in \eqref{Sigma_mugmug} is nonzero for $g\in[G]$, and
        \item\label{inver-cond3} $H_0$ defined in \eqref{Hk0} is invertible.
    \end{enumerate}
\end{assumption}
Let $\widehat{\mu}_g(\theta)= \widehat{\mu}_g(\theta,P_g(\theta))$, Assumption \ref{assu:invertibility2}\ref{inver-cond1} guarantees that the partial derivative $\partial\widehat\mu_g(\theta)/\partial\theta$ is well defined for $\theta\in\Theta$. $\Sigma_{\mu_g\mu_g0}$ in part \ref{inver-cond2} is the Hessian matrix for $\mu_g$ and $H_0$ in part \ref{inver-cond3} is the concentrated matrix for $\theta$, both at the true value. Therefore, these two assumptions are used in establishing asymptotic normality of the common parameters, making sure the asymptotic variance is bounded.
\begin{proof}[Proof of Theorem \ref{thm:post_classification_estimation}]
The optimality conditions for $\widehat\theta$ are characterized by the following equations:
\begin{align}
    \text{(FOC for } \mu_g\text{):} \quad 
    \frac{1}{n} \sum_{i \in [n]} \frac{\partial \psi(w_{ig}, P_g(\widehat\zeta_g); \widehat\zeta_g)}{\partial \mu_g} 
    &= 0, \label{FOC_mu} \\
    \text{(FOC for } \theta\text{):} \quad 
    \frac{1}{nG} \sum_{g \in [G]} \sum_{i \in [n]} \frac{\partial \psi(w_{ig}, P_g(\widehat\zeta_g); \widehat\zeta_g)}{\partial \theta} 
    &= 0, \label{FOC_delta} \\
    \text{(Definition of } P_g(\widehat\zeta_g)\text{):} \quad 
    P_g(\widehat\zeta_g) - \Gamma_g(\widehat\zeta_g; P_g(\widehat\zeta_g)) 
    &= 0. \label{def_P}
\end{align}
By \eqref{FOC_mu}--\eqref{def_P}, we first characterize $\partial\widehat{\mu}_g(\theta)/\partial\theta$. Taking a derivative of \eqref{FOC_mu} with respect to $\theta$, we have 
\[
\begin{aligned}
        0=\frac{1}{n}\sum_{i\in[n]}\Bigg[&\frac{\partial^2\psi(w_{ig}, P_g(\widehat\mu_g(\theta),\theta);\widehat\mu_g(\theta),\theta)}{\partial\mu_g^2}\\&+\frac{\partial^2\psi(w_{ig}, P_g(\widehat\mu_g(\theta),\theta);\widehat\mu_g(\theta),\theta)}{\partial\mu_g\partial P_g'}
        \times\frac{\partial P_g(\widehat\mu_g(\theta),\theta)}{\partial\mu_g'} \Bigg]\frac{\partial\widehat{\mu}_g(\theta)}{\partial\theta'}\\
        +\frac{1}{n}\sum_{i\in[n]}&\Bigg[\frac{\partial^2\psi(w_{ig}, P_g(\widehat\mu_g(\theta),\theta);\widehat\mu_g(\theta),\theta)}{\partial\mu_g\partial\theta'}\\&+\frac{\partial^2\psi(w_{ig}, P_g(\widehat\mu_g(\theta),\theta);\widehat\mu_g(\theta),\theta)}{\partial\mu_g\partial P_g'}
        \times\frac{\partial P_g(\widehat\mu_g(\theta),\theta)}{\partial\theta'} \Bigg].
    \end{aligned}
\]
Therefore, by Assumption \ref{assu:invertibility2}, we have
\begin{equation}
\label{partial_derivatives}
\begin{aligned}
   \frac{\partial\widehat{\mu}_g(\theta)}{\partial\theta'} =&-\Bigg\{\frac{1}{n}\sum_{i\in[n]}\Bigg[\frac{\partial^2\psi(w_{ig}, P_g(\widehat\mu_g(\theta),\theta);\widehat\mu_g(\theta),\theta)}{\partial\mu_g^2}\\&\quad\quad+\frac{\partial^2\psi(w_{ig}, P_g(\widehat\mu_g(\theta),\theta);\widehat\mu_g(\theta),\theta)}{\partial\mu_g\partial P_g'}\frac{\partial P_g(\widehat\mu_g(\theta),\theta)}{\partial\mu_g'} \Bigg]\Bigg\}^{-1}\\
        &\quad\quad\times\frac{1}{n}\sum_{i\in[n]}\Bigg[\frac{\partial^2\psi(w_{ig}, P_g(\widehat\mu_g(\theta),\theta);\widehat\mu_g(\theta),\theta)}{\partial\mu_g\partial\theta'}\\&\quad\quad\quad\quad+\frac{\partial^2\psi(w_{ig}, P_g(\widehat\mu_g(\theta),\theta);\widehat\mu_g(\theta),\theta)}{\partial\mu_g\partial P_g'}\frac{\partial P_g(\widehat\mu_g(\theta),\theta)}{\partial\theta'} \Bigg].
    \end{aligned}
\end{equation}

Define 
\begin{align*}
    H_{ig}&(\widehat{\mu}_g(\theta), \theta)\\
    &= \frac{\partial^2 \psi(w_{ig}, P_g(\widehat{\mu}_g(\theta), \theta); \widehat{\mu}_g(\theta), \theta)}
             {\partial \theta \partial \theta'} 
     + \frac{\partial^2 \psi(w_{ig}, P_g(\widehat{\mu}_g(\theta), \theta); \widehat{\mu}_g(\theta), \theta)}
                  {\partial \theta \partial \mu_g'}
                  \frac{\partial \widehat{\mu}_g(\theta)}{\partial \theta'} \\
    &\quad + \frac{\partial^2 \psi(w_{ig}, P_g(\widehat{\mu}_g(\theta), \theta); \widehat{\mu}_g(\theta), \theta)}
                  {\partial P_g'}\Bigg[
                      \frac{\partial P_g(\widehat{\mu}_g(\theta), \theta)}{\partial \theta'}
                      + \frac{\partial P_g(\widehat{\mu}_g(\theta), \theta)}{\partial \mu_g}
                          \frac{\partial \widehat{\mu}_g(\theta)}{\partial \theta'}
                  \Bigg]
\end{align*}
and
\begin{equation}
\label{H_k}
H(\widehat\mu(\theta),\theta)=\frac{1}{nG}\sum_{g\in [G]}\sum_{i\in[n]} H_{ig}(\widehat\mu_g(\theta),\theta).
\end{equation}
We define the probability limit of \eqref{H_k} as
\begin{equation}
    \label{Hk0}H_{0}=\underset{n\to\infty}{\plim}~H(\widehat\mu(\theta_{0}),\theta_{0}).
\end{equation}
Then a Taylor expansion of FOC for $\theta$, \eqref{FOC_delta}, with respect to $\theta$ around $\theta_{0}$ yields 
\begin{equation}
\label{eq:taylor-psi}
    \begin{aligned}
        -\frac{1}{nG}\sum_{g\in [G]}&\sum_{i\in[n]}\frac{\partial \psi(w_{ig}, P_g(\widehat\mu_{g}(\theta_{0});\widehat\mu_g(\theta_{0}),\theta_{0})}{\partial\theta}\\
        =\ &H(\widehat\mu_g(\breve\theta),\breve\theta)(\widehat\theta-\theta_{0})=H_0(\widehat\theta-\theta_0)+o_p((nG)^{-1/2}).
    \end{aligned}
\end{equation}

From \eqref{eq:taylor-psi}, we obtain an asymptotic representation for $\sqrt{nG}(\widehat\theta-\theta_{0})$:
\begin{align*}
 &\sqrt{nG}\big(\widehat\theta-\theta_{0}\big)\\
 =&-H_{0}^{-1}\Bigg\{\underbrace{\frac{1}{\sqrt{nG}}\sum_{g\in [G]}\sum _{i\in[n]}\frac{\partial \psi(w_{ig}, P_g(\widehat\mu_g(\theta_{0}),\theta_{0});\widehat\mu_g(\theta_{0}),\theta_{0})}{\partial\theta}}_{(I)}\Bigg\}+o_p(1).
\end{align*}
We cannot directly apply a central limit theorem to the above term (I) because of the presence of $\widehat\mu_g(\theta_{0})$. 
Instead, we expand it with respect to $\mu_g$ around $\mu_{g0}$ by a third-order Taylor expansion. Throughout the expansion, we write $\psi_{ig0}:=\psi(w_{ig},P_{g0};\zeta_{g0})$ and let the subscript ``$0$'' denote evaluation at the true values.
\[
        \begin{aligned}
            &\underbrace{\frac{1}{\sqrt{nG}}\sum_{g\in [G]}\sum_{i\in[n]}\frac{\partial \psi_{ig0}}{\partial\theta}+\frac{1}{\sqrt{nG}}\sum_{g\in [G]}\sum_{i\in[n]}\left[\frac{\partial^2\psi_{ig0}}{\partial\theta\partial\mu_g}+\frac{\partial^2\psi_{ig0}}{\partial\theta\partial P_g'}\frac{\partial P_{g0}}{\partial \mu_g}\right]\big(\widehat\mu_g(\theta_{0})-\mu_{g0}\big)}_{R_1}\\
           & +\underbrace{\frac12\frac{1}{\sqrt{nG}}\sum_{g\in [G]}\sum_{i\in[n]}\left[\frac{\partial^3 \psi_{ig0}}{\partial\theta\partial\mu_g^2}+\frac{\partial^3\psi_{ig0}}{\partial\theta\partial P_g'\partial\mu_g}\frac{\partial P_{g0}}{\partial \mu_g}+\frac{\partial^2\psi_{ig0}}{\partial\theta\partial P_g'}\frac{\partial^2P_{g0}}{\partial\mu_g^2}\right]\big(\widehat\mu_g(\theta_{0})-\mu_{g0}\big)^2}_{R_2}\\
           &+\underbrace{\frac16\frac{1}{\sqrt{nG}}\sum_{g\in [G]}\sum_{i\in[n]}\mathcal{T}_{ig0}(\widehat\mu_g(\theta_{0})-\mu_{g0})^3}_{R_3}+\text{higher order terms},
        \end{aligned}
\]
where $\mathcal{T}_{ig0}$ is the third-order derivative. We analyze $R_1$, $R_2$, and $R_3$ in turn.

For $R_1$, define
\begin{align}
\label{Sigma_mugmug}
    \Sigma_{\mu_g\mu_g0}&=\underset{n\to\infty}{\plim}~\frac{1}{n}\sum_{i\in[n]}\left(\frac{\partial^2\psi_{ig0}}{\partial\mu_g^2}+\frac{\partial^2\psi_{ig0}}{\partial \mu_g\partial P_g'}\frac{\partial P_{g0}}{\partial\mu_g}\right), \\
\label{Sigma_deltakmug}
\Sigma_{\theta\mu_g0}&=\underset{n\to\infty}{\plim}~\frac{1}{n}\sum_{i\in[n]}\left(\frac{\partial^2\psi_{ig0}}{\partial\theta\partial\mu_g}+\frac{\partial^2\psi_{ig0}}{\partial\theta\partial P_g'}\frac{\partial P_{g0}}{\partial \mu_g}\right), \\
\label{Omega1}
    \Omega_{\mu_g\mu_g0}&=\mathbb{E}\left[\frac{1}{n}\sum_{i\in[n]}\left(\frac{\partial\psi_{ig0}}{\partial\mu_g}\right)^2\right], \\
\label{Omega2}
    \Omega_{\mu_g\theta0}&= \mathbb{E}\left[\frac{1}{n}\sum_{i\in[n]}\frac{\partial\psi_{ig0}}{\partial\mu_g} \times \frac{\partial\psi_{ig0}}{\partial\theta}\right],\text{ and } \\
\label{Omega3}
\Omega_{\theta\theta0}&= \mathbb{E}\left[\frac{1}{n}\sum_{i\in[n]}\frac{\partial\psi_{ig0}}{\partial\theta}\frac{\partial\psi_{ig0}}{\partial\theta'}\right].
\end{align}
An argument parallel to the derivations in Lemma \ref{lem:dev_bounds} yields
\begin{equation}
\label{mu-expansion}
    \widehat{\mu}_g(\theta_{0})-\mu_{g0}=-\Sigma_{\mu_g\mu_g0}^{-1}\times\frac{1}{n}\sum_{i\in[n]}\frac{\partial \psi_{ig0}}{\partial \mu_g}+O_p\left(\frac{\log(Gn)}{n}\right),
\end{equation}
where we use Assumption \ref{assu:invertibility2}\ref{inver-cond2}. Then we  substitute it into $R_1$,  
\[
R_1=\frac{1}{\sqrt{nG}}\sum_{g\in [G]}\sum_{i\in[n]}\left(\frac{\partial \psi_{ig0}}{\partial\theta}-\Sigma_{\theta\mu_g0}\Sigma_{\mu_g\mu_g0}^{-1}\frac{\partial \psi_{ig0}}{\partial \mu_g}\right)+o_p(1).
\]
By the Lindeberg-Feller central limit theorem \citet[Theorem 5.6]{white2014asymptotic}, we have
\[
R_1\stackrel{d}{\to}N(0,\Omega_{0}),
\]
where the asymptotic variance is  
\begin{equation}
    \label{omegak0}
    \Omega_{0}
    = \lim_{n\to\infty}
      \left\{
        \begin{aligned}
            &\Sigma_{\theta\theta0} - \frac{1}{G} \sum_{g \in [G]}
                \Sigma_{\mu_g\mu_g0}^{-1} \\
                \times&\Big(
                    \Sigma_{\theta\mu_g0} \Omega_{\mu_g\theta0}
                    + \Omega_{\mu_g\theta0}' \Sigma_{\theta\mu_g0}'
                    - \Sigma_{\mu_g\mu_g0}^{-1}
                        \Sigma_{\theta\mu_g0}
                        \Omega_{\theta\theta0}
                        \Sigma_{\theta\mu_g0}'
                \Big)
        \end{aligned}
      \right\}.
\end{equation}

Turning to $R_2$, notice that by \eqref{mu-expansion}, we have $$\frac{1}{\sqrt{n}}\sum_{i\in[n]}\frac{\partial\psi_{ig0}}{\partial\mu_g}=O_p(1),$$  and then
\begin{align*}
 [\sqrt{n}(\widehat{\mu}_g(\theta_{0})-\mu_{g0})]^2=\ &\Sigma_{\mu_g\mu_g0}^{-2}\times\frac{1}{n}\left(\sum_{i\in[n]}\frac{\partial\psi_{ig0}}{\partial\mu_g}\right)^2\\&-O_p\left(\frac{\log(Gn)}{\sqrt{n}}\right)\times\frac{2}{\sqrt{n}}\sum_{i\in[n]}\frac{\partial\psi_{ig0}}{\partial\mu_g}+O_p\left(\frac{(\log(Gn))^2}{n}\right)\\
 =\ &\Sigma_{\mu_g\mu_g0}^{-2}\times\frac{1}{n}\left(\sum_{i\in[n]}\frac{\partial\psi_{ig0}}{\partial\mu_g}\right)^2+o_p(1)\\
 =\ &\Sigma_{\mu_g\mu_g0}^{-2}\Omega_{\mu_g\mu_g0}+o_p(1),
\end{align*}
where the last equality holds by a weak law of large numbers (WLLN) and the fact that 
\begin{align*}
    \mathbb{E}\left[\frac{1}{n}\left(\sum_{i \in [n]} \frac{\partial \psi_{ig0}}{\partial \mu_g}\right)^2\right]
    &= \mathbb{E}\left[\frac{1}{n} \left(\sum_{i \in [n]} \left(\frac{\partial \psi_{ig0}}{\partial \mu_g}\right)^2 + \sum_{i 
\neq j} \frac{\partial \psi_{ig0}}{\partial \mu_g} \cdot \frac{\partial \psi_{jg0}}{\partial \mu_g}\right)\right] \\
    &= \mathbb{E}\left[\frac{1}{n} \sum_{i \in [n]} \left(\frac{\partial \psi_{ig0}}{\partial \mu_g}\right)^2\right] \\
    &= \Omega_{\mu_g\mu_g0}.
\end{align*}
Define 
\[
  \Sigma_{\theta\mu_g\mu_g0} =\underset{n\to\infty}{\plim}~\frac{1}{n}\sum_{i\in[n]}\left(\frac{\partial^3 \psi_{ig0}}{\partial\theta\partial\mu_g^2}+\frac{\partial^3\psi_{ig0}}{\partial\theta\partial P_g'\partial\mu_g}\frac{\partial P_{g0}}{\partial \mu_g}+\frac{\partial^3\psi_{ig0}}{\partial\theta\partial P_g'}\frac{\partial^2P_{g0}}{\partial\mu_g^2}\right), 
\]
and substitute $\widehat\mu_g(\theta_{0})-\mu_{g0}$ into $R_2$,
we have 
\[
    \begin{aligned}
        R_2&=\frac12\frac{1}{\sqrt{nG}}\sum_{g\in [G]}\Sigma_{\theta\mu_g\mu_g0}\Sigma_{\mu_g\mu_g0}^{-2}\Omega_{\mu_g\mu_g0}+o_p(1)\\
        &=\frac12\sqrt{\frac{G}{n}}\frac{1}{G}\sum_{g\in [G]}\Sigma_{\theta\mu_g\mu_g0}\Sigma_{\mu_g\mu_g0}^{-2}\Omega_{\mu_g\mu_g0}+o_p(1)\\
        &\stackrel{p}\to B_{0},
    \end{aligned}
\]
where $B_{0}$  is defined as
\begin{equation}
\label{Bk0}
    B_{0}=\underset{n\to\infty}{\plim}~\frac12\sqrt{\frac{G}{n}}\frac{1}{G}\sum_{g\in [G]}\Sigma_{\theta\mu_g\mu_g0}\Sigma_{\mu_g\mu_g0}^{-2}\Omega_{\mu_g\mu_g0}
\end{equation}
in view of $G/n\to c\in(0,\infty)$.

For $R_3$, recall that $[\sqrt{n}(\widehat\mu_g(\theta_{0})-\mu_{g0})]^2=O_p(1)$. Then define
$$\Sigma_{\theta\mu_g\mu_g\mu_g0}=\underset{n\to\infty}{\plim}~\frac{1}{n}\sum_{i\in[n]}\mathcal{T}_{ig0},$$ which is finite as the third derivative of the logistic CDF $\Lambda(x)$ is bounded when $|x|<\infty$.  We have
\[
    \begin{aligned}
        R_3=\ &\frac16\frac{1}{\sqrt{nG}}\sum_{g\in [G]}\sum_{i\in[n]}\mathcal{T}_{ig0}(\widehat\mu_g(\theta_{0})-\mu_{g0})^3\\
        =\ &\frac16\frac{1}{\sqrt{nG}}\sum_{g\in [G]}\left(\frac{1}{n}\sum_{i\in[n]}\mathcal{T}_{ig0}\right)[\sqrt{n}(\widehat\mu_g(\theta_{0})-\mu_{g0})]^2(\widehat\mu_g(\theta_{0})-\mu_{g0})\\
        =\ &\frac16\frac{1}{\sqrt{nG}}\sum_{g\in [G]}\Sigma_{\theta\mu_g\mu_g\mu_g0}\times o_p(1)\times O_p(1)+o_p(1)\\
        =\ &\frac16\left(\sqrt{\frac{G}{n}}\frac{1}{G}\sum_{g\in [G]}\Sigma_{\theta\mu_g\mu_g\mu_g0}\right)\times o_p(1)+o_p(1)\\
        =\ &o_p(1),
    \end{aligned}
\]
as $G/n\to c\in(0,\infty)$. The terms higher order than $R_3$ are $o_p(1)$ by similar arguments.  We have proved part \ref{thm:post_classification_estimation-a} of the theorem.

To validate the parametric bootstrap inference procedure, it suffices to verify Assumptions 1–6 in \citet{higgins2024bootstrap}. First, since the logistic link function is infinitely differentiable, the true parameters lie in the interior of a compact parameter space, and the $\epsilon_{ig}$'s are i.i.d., our setting readily satisfies Assumptions 1–4 of \citet{higgins2024bootstrap}. Second, since we assume $G/n \to c \in (0,\infty)$, their Assumption 5 also holds in our setting. Finally, Assumptions \ref{assu:invertibility_assumption1} and \ref{assu:invertibility2} are equivalent to their Assumption 6.

To see why these invertibility conditions hold for our model, let $\eta_{ig} = \overline{P}_{ig}\overline{\beta} + X_{ig}'\beta + \mu_g$ be the linear index and $\Lambda'_{ig} = \Lambda(\eta_{ig})(1-\Lambda(\eta_{ig}))$.  For the logistic log-likelihood, $$\partial\psi_{ig}/\partial\theta = -(Y_{ig} - \Lambda(\eta_{ig}))\,\partial\eta_{ig}/\partial\theta,$$ and the Hessian is $$\Lambda'_{ig}\,(\partial\eta_{ig}/\partial\theta)(\partial\eta_{ig}/\partial\theta')'$$ plus derivatives with respect to the CCPs through the fixed point mapping.  Since $\Lambda'_{ig} \in (0,1/4]$ and $\|X_{ig}\|$ is bounded (Assumption~\ref{npl}\ref{npl-i}), the Hessian is uniformly bounded above and, under the identification condition in Assumption~\ref{npl}\ref{npl-iv}, bounded away from zero.  This allows us to complete the proof of part \ref{thm:post_classification_estimation-b} by adapting the arguments of \citet{higgins2024bootstrap} to our setting.
\end{proof}

\subsection{Proof of Theorem \ref{thm:classification_consistency}}\label{appendix:classification}
Following the same argument as in the proof of  Theorem 2.2 in \citet{su2016identifying}, the results in this theorem hold if  $P_0$ is known. Now we replace it  with $\widetilde{P}$, where $\|\widetilde{P}_g-P_{g0}\|_2=O_p(\sqrt{\log(Gn)})$ uniformly for $g\in[G]$. This substitution does not affect the result. For each $i$ and each $g$, we have
\begin{equation*}
    \sum_{i\in[n]}\frac{\partial{\psi_{ig}(w_{ig},P_g;\zeta_g)}}{\partial P_g}=\sum_{i\in[n]}\left(\frac{\overline\beta_gF^g_{i1}(Y_{ig}-P_{ig})}{N_i^g},\dots,\frac{\overline\beta_gF^g_{in}(Y_{ig}-P_{ig})}{N_i^g}\right)'.
\end{equation*}
Therefore, 
\begin{equation}
\label{eq:thm2-fact1}
    \left\|\sum_{i\in[n]}\frac{\partial{\psi_{ig}(w_{ig},P_g;\zeta_g)}}{\partial P_g}\right\|_\infty\leq2|\overline\beta_g|\max_{j\in[n]}\sum_{i\in[n]}\frac{F^g_{ij}}{N_i^g}\leq C,
\end{equation}
where the first inequality uses the fact that $|Y_{ig}-P_{ig}|\leq2$ and the second inequality holds by Assumption \ref{npl}\ref{npl-v}. We have
\begin{align*}
|Q_{nG}&(\theta;\widetilde{P})-Q_{nG}(\theta;P_0)|\\
&\leq\frac{1}{nG}\sum_{g\in[G]}\Big\|\sum_{i\in[n]}\frac{\partial{\psi_{ig}(w_{ig},\breve{P}_g;\theta_g,\widehat{\mu}_g(\theta_g,\breve{P}_g))}}{\partial P_g}\Big\|_2\|\widetilde{P}_g-P_{g0}\|_2,\\
&=\frac{1}{nG}\times G\times O_p(\sqrt{n}) \times O_p\left(\sqrt{\log(Gn)}\right)\\
&=O_p\Big(\sqrt{\frac{\log(Gn)}{n}}\Big),
\end{align*}
where the first inequality uses a Taylor expansion and the Cauchy--Schwarz (CS) inequality and the first equality holds by \eqref{eq:thm2-fact1}.
Then, the classification consistency is preserved under this perturbation, as we have verified by inspecting the proofs in \citet{su2016identifying}.

\subsection{Proof of Proposition \ref{prop:K-selection}}\label{appendix:k-selection}
The proof follows the same idea as the proof of Theorem 2.6 in \citet{su2016identifying}; we omit the details.

\subsection{Proof of Corollary \ref{cor:ape_asymptotics}}\label{appendix:ape_corollary}
\begin{proof}[Proof of Corollary \ref{cor:ape_asymptotics}]
The APE functional $\delta(\theta_k, P_k)$ in~\eqref{eq:ape} is continuously differentiable in $\theta_k$ with gradient $\nabla\delta^0_{k,p}$ given by~\eqref{eq:ape_gradient}, and continuous in $P_k$.
Theorem~\ref{thm:classification_consistency} implies $\widetilde{C}_k = C_{k0}$ with probability approaching one, so we may treat the cluster partition as known for the asymptotic analysis.
Theorem~\ref{thm:post_classification_estimation}\ref{thm:post_classification_estimation-b} establishes $$\sqrt{nG_k}(\widehat{\theta}^*_{\widetilde{C}_k} - \theta_{k0}) \stackrel{d}{\to} N(0, H_{0k}^{-1}\Omega_{0k}H_{0k}^{-1}),$$ and Proposition~\ref{prop:first_step_npl} ensures $\widehat{P}_{\widetilde{C}_k}$ is uniformly consistent for $P_{k0}$.
By the delta method,
\[
\sqrt{nG_k}\big(\widehat{\delta}_{k,p}^{\,*} - \delta^0_{k,p}\big)
= (\nabla\delta^0_{k,p})' \cdot \sqrt{nG_k}\big(\widehat{\theta}^*_{\widetilde{C}_k} - \theta_{k0}\big) + o_p(1)
\stackrel{d}{\to} N\big(0, (\nabla\delta^0_{k,p})' H_{0k}^{-1}\Omega_{0k}H_{0k}^{-1} \nabla\delta^0_{k,p}\big),
\]
which proves part~\ref{cor:ape_asymptotics-a}.
For part~\ref{cor:ape_asymptotics-b}, Theorem~\ref{thm:post_classification_estimation}\ref{thm:post_classification_estimation-c} shows that the bootstrap distribution of $\sqrt{nG_k}a'(\widehat{\theta}^b - \widehat{\theta})$ consistently estimates that of $\sqrt{nG_k}a'(\widehat{\theta} - \theta_0)$ for any non-random vector $a$.
Setting $a = \nabla\delta^0_{k,p}$ and applying the delta method to the bootstrap statistic yields
$$\sqrt{nG_k}(\widehat{\delta}_{k,p}^{\,b} - \widehat{\delta}_{k,p}^{\,*}) = (\nabla\delta^0_{k,p})' \cdot \sqrt{nG_k}(\widehat{\theta}^b - \widehat{\theta}^*) + o_p(1).$$
Taking the supremum over $a\in\mathbb{R}$ conditional on $\boldsymbol{W}$ gives the uniform approximation stated in part~\ref{cor:ape_asymptotics-b}, which implies that the bootstrap confidence interval attains the nominal coverage probability asymptotically.
\end{proof}

\section{Proofs of Lemmas}\label{appendix-lemmas}
\renewcommand{\theequation}{B.\arabic{equation}}
\setcounter{equation}{0}
In this Appendix, we present the proofs of Lemmas \ref{lem:contraction_mapping} and \ref{lem:dev_bounds} in Appendix \ref{appendix}.
\subsection{Proof of Lemma \ref{lem:contraction_mapping}}
The logistic CDF $\Lambda(x)=1/(1+\mathrm e^{-x})$ satisfies
$\frac{d}{d x}
\Lambda(x)=\Lambda(x)[1-\Lambda(x)]$.
Then, the Jacobian matrix
of $\Gamma(\zeta,P)$ with respect to $P$ , denoted by $J_n(\zeta,P)$, satisfies 
\[
    [J_n(\zeta,P)]_{ij}=\frac{\overline\beta F_{ij}}{N_i} \Lambda\Big(\mu+X_{i}'\beta+\overline{P}_i\overline{\beta}\Big)\Big[1-\Lambda\Big(\mu+X_{i}'\beta+\overline{P}_i\overline{\beta}\Big)\Big], \quad \forall\, 1\leq i,j\leq n.
\]
Therefore,  we have 
\begin{equation}
\label{eq:lem1-fact1}
    \Lambda\Big(\mu+X_{i}'\beta+\overline{P}_i\overline{\beta}\Big)\Big[1-\Lambda\Big(\mu+X_{i}'\beta+\overline{P}_i\overline{\beta}\Big)\Big]\in(0,1/4],
\end{equation}
under Assumption \ref{npl}\ref{npl-i} and the condition $\|\zeta\|_{2}<\infty$.
Recalling the definition of the social interaction matrix $F$, we have
\begin{equation}
\label{eq:lem1-fact2}
    F_{ij}\in\{0,1\} \text{ and } \frac{1}{N_i}\sum_{j\in[n]}F_{ij}\leq 1.
\end{equation}
Combining \eqref{eq:lem1-fact1} and \eqref{eq:lem1-fact2}, we have 
\begin{align*}
\left\| {J}_{n}(\zeta,{P})\right\|_{\infty} & =
\max_i \sum_{j=1}^{n}\left|[J_{n}(\zeta,P)]_{ij}\right|\ \\
 & = \max_i  \ |\overline\beta|\sum_{j\in[n]}\frac{F_{ij}}{N_i}\Lambda\Big(\mu+X_{i}'\beta+\overline{P}_i\overline{\beta}\Big)\Big[1-\Lambda\Big(\mu+X_{i}'\beta+\overline{P}_i\overline{\beta}\Big)\Big]\\
& \leq  \max_i  \frac{|\overline\beta|}{4}\sum_{j\in[n]}\frac{F_{ij}}{N_i}\leq\frac{|\overline\beta|}{4}\in(0,1)
\end{align*}
where the last inequality uses the fact $|\overline \beta| \in (0,4)$.
Therefore $\Gamma(\zeta,{P})$ is a contraction mapping. Then by the Banach fixed-point theorem, $P(\zeta)$ exists and is unique.
Moreover,  $I_n-J_n(\zeta,P)$ is invertible as it is strictly diagonally dominant. Hence,  the implicit function theorem implies
$$\frac{\partial P(\zeta)}{\partial\zeta'}=\Big[I_n-J_n(\zeta,P(\zeta))\Big]^{-1}\frac{\partial\Lambda(\zeta,P(\zeta))}{\partial\zeta'}.$$  By the Gershgorin circle theorem  \citep[Theorem 6.1.1]{horn2012matrix}, the smallest eigenvalue of $I_n-J_n(\zeta,P(\zeta))$ is  strictly positive, hence
$$
\label{lemma-existence-fact1}
\left\|\Big[I_n-J_n(\zeta,P(\zeta))\Big]^{-1}\right\|_2\leq C.
$$
Since all elements in $\frac{\partial\Lambda(\zeta,P(\zeta))}{\partial\zeta'}$ are finite, we have 
$$
\label{lemma-existence-fact2}
\left\|\frac{\partial\Lambda(\zeta,P(\zeta))}{\partial\zeta'}\right\|_2\leq \left\|\frac{\partial\Lambda(\zeta,P(\zeta))}{\partial\zeta'}\right\|_F\leq C\sqrt{n}.
$$
Combining the two bounds, we have
$$
\left\| \frac{\partial P(\zeta)}{\partial\zeta}\right\|_2\leq  \left\| \Big[I_n-J_n(\zeta,P(\zeta))\Big]^{-1}\right\|_2\left\| \frac{\partial\Lambda(\zeta,P(\zeta))}{\partial\zeta'}\right\|_2\leq C\sqrt{n}. 
$$

\subsection{Proof of Lemma \ref{lem:dev_bounds}}
For the first deviation bound, because $\zeta$  is finite-dimensional, we may, without loss of generality, consider the derivative with respect to $\mu$, namely $n^{-1}\sum_{i\in[n]}\partial\psi(w_i, 
P_0;\zeta_0)/\partial\mu$. By a direct calculation, we have $$\frac{1}{n}\sum_{i\in[n]}\frac{\partial\psi(w_i, P_0;\zeta_0)}{\partial\mu}=\frac{1}{n}\sum_{i\in[n]}(Y_i-P_{i0}).$$ Since $|Y_i|\leq1$, Hoeffding's inequality \citep{hoeffding1963probability} implies that $$\Pr\Big(\Big|\frac{1}{n}\sum_{i\in[n]}(Y_i-P_{i0})\Big|\geq t\Big)\leq 2\exp\left(-2nt^2\right).$$
The result follows by setting $t=\sqrt{\log(Gn)/n}$.

For the second deviation bound, notice that $$\left|\frac{1}{n}\sum_{i\in[n]}\{\psi(w_i,P;\zeta)-\mathbb{E}[\psi(w_i,P;\zeta)]\}\right|=\left|\frac{1}{n}\sum_{i\in[n]}(Y_{i}-P_{i0})\left(X_{i}'\beta+\mu+\overline{P}_i\overline{\beta}\right)\right|.
$$
By Assumption \ref{npl}, we have $$|(Y_i-P_{i0})(X_{i}'\beta+\mu+\overline{P}_i\overline{\beta})|\leq |X_{i}'\beta+\mu+\overline{P}_i\overline{\beta}|<C$$ uniformly over $(\zeta,P)\in\mathcal{A}\times\Theta\times\mathcal{P}_n$. Then the second bound follows by applying an extended  Hoeffding inequality for empirical processes \citep[Theorem 12.1]{boucheron_concentration_2013}.

\end{document}